\documentclass[submission,copyright,creativecommons]{eptcs}
\providecommand{\event}{GCM-2020 Post Proceedings}

\usepackage[utf8]{inputenc}
\usepackage[english]{babel}

\usepackage[titletoc]{appendix}
\usepackage{stmaryrd}
\usepackage{amssymb}
\usepackage{amsmath}
\usepackage{graphicx}
\usepackage{enumerate}
\usepackage{enumitem}
\usepackage{multicol}
\usepackage{subcaption}
\usepackage{amsthm}
\usepackage[normalem]{ulem}
\usepackage{array}

\theoremstyle{plain}
\newtheorem{lemma}{Lemma}
\newtheorem{theorem}{Theorem}
\newtheorem{proposition}{Proposition}

\theoremstyle{definition}
\newtheorem{transform}{Transformation}
\newtheorem{definition}{Definition}

\theoremstyle{remark}
\newtheorem{example}{Example}
\newtheorem{remark}{Remark}

\usepackage{tikz}
\usetikzlibrary{patterns,arrows, topaths, calc, positioning}
\tikzset{>=stealth}
\tikzstyle{node} = [circle, minimum size = 1.4mm, inner sep = 0mm, color=black, fill]
\tikzstyle{hyperedge} = [rectangle, minimum width = 5mm, minimum height = 5mm, draw, inner sep = 0mm]
\tikzstyle{HG} = [align = center]
\tikzstyle{circledge} = [circle, minimum size = 7mm, inner sep = 0mm, color=black, draw]

\pgfdeclarepatternformonly{NELines}{\pgfqpoint{-1pt}{-1pt}}{\pgfqpoint{4pt}{4pt}}{\pgfqpoint{3pt}{3pt}}%
{
	\pgfsetlinewidth{0.4pt}
	\pgfpathmoveto{\pgfqpoint{0pt}{0pt}}
	\pgfpathlineto{\pgfqpoint{1pt}{1pt}}
	\pgfpathmoveto{\pgfqpoint{2pt}{2pt}}
	\pgfpathlineto{\pgfqpoint{3.1pt}{3.1pt}}
	\pgfusepath{stroke}
}

\newcommand\scale{1.7}

\title{Weak Greibach Normal Form for Hyperedge Replacement Grammars}
\author{Tikhon Pshenitsyn
	\institute{Department of Mathematical Logic and Theory of Algorithms\\Faculty of Mathematics and Mechanics\\
		Lomonosov Moscow State University
		\\GSP-1, Leninskie Gory, Moscow, 119991, Russian Federation}
	\thanks{The study was funded by RFBR, project number 20-01-00670.}
	\email{tpshenitsyn@lpcs.math.msu.su}
}
\def\authorrunning{T. Pshenitsyn}
\def\titlerunning{Weak Greibach Normal Form for HRGs}

\begin{document}
	\maketitle
\begin{abstract}
	It is known that hyperedge replacement grammars are similar to string context-free grammars in the sense of definitions and properties. Therefore, we expect that there is a generalization of the well-known Greibach normal form from string grammars to hypergraph grammars. Such generalized normal forms are presented in several papers; however, they do not cover a large class of hypergraph languages (e.g.~languages consisting of star graphs). In this paper, we introduce a weak Greibach normal form, whose definition corresponds to the lexicalized normal form for string grammars, and prove that every context-free hypergraph language (with nonsubstantial exceptions) can be generated by a grammar in this normal form. The proof presented in this paper generalizes a corresponding one for string grammars with a few more technicalities. 
\end{abstract}
\section{Introduction}
	Extensions of formal grammars and languages from strings to graphs are considered in a wide variety of works. The resulting formalisms called graph grammars generate graph languages by means of productions: each production is of the form $A\to G$; it allows one to replace a part of a graph labeled by $A$ with the graph $G$ if certain conditions are satisfied. An overview on graph grammars can be found in ``Handbook on Graph Grammars and Computing by Graph Transformation'' \cite{Rozenberg97}.
	
	In this paper, we focus on a particular approach called hyperedge replacement grammar (HRG). An overview on HRG can be found in the book \cite{Habel92} or in a chapter of the handbook mentioned above \cite{Drewes97}. In brief, a hyperedge replacement grammar contains productions that allow one to replace an edge with a certain label by a hypergraph; the rest is similar to string context-free grammar definitions. HRGs appeared in the seventies, and they became popular from theoretical and practical points of view; particularly, they can be used in machine translation or programming (see, for instance, \cite{Jones12}). Many structural properties of HRGs have been studied. Nicely, many theorems regarding context-free string grammars (CFGs) can be extended to HRGs in a natural way; besides, it is often the case that proofs for CFGs can be generalized to HRGs so there is no need to invent anything new.
	
	For context-free grammars, there is a so-called Greibach normal form, which is helpful in a number of investigations, e.g.~to connect CFGs with categorial grammars (see \cite{BarHillel60}). Each production in a grammar in the Greibach normal form has to begin with a terminal symbol and proceed with at most two nonterminal ones. This definition can be weakened: one says that a grammar is in the weak Greibach normal form (it is often called \emph{lexicalized}) if there is exactly one terminal symbol in the right-hand side of each production. Obviously, the same definition can be introduced for graphs. Namely, an HRG is said to be in the weak Greibach normal form if for each production $A\to H$ the hypergraph $H$ has exactly one edge labeled by a terminal symbol.
	
	The main objective of this paper is to present a precise class of graph languages that can be generated by HRGs in the weak Greibach normal form. In order to do this we generalize the proof of existence of the Greibach normal form for context-free grammars  taken from \cite{Aho72} in a straightforward way. However, several differences exist. Firstly, there is a nonsubstantial class of hypergraph context-free languages that cannot be generated by grammars in the weak Greibach normal form (a problem arises because of isolated nodes). Secondly, there is a trick in the proof for CFGs that exploits the string nature of grammars, so it is hard to generalize it to hypergraphs. This led us to a large and technically heavy proof. Finally, there are bad news regarding algorithmical complexity: conversion of an HRG into an equivalent one in the weak Greibach normal form cannot be done even in exponential time.
	
	In Sect. \ref{RW} we compare our contribution with other studies related to the Greibach normal form for graph grammars. In Sect. \ref{s_defn} formal definitions related to hypergraphs and HRGs are presented. In Sect. \ref{s_GNF} the weak Greibach normal form is introduced, related issues are discussed. Sect. \ref{s_proof} is devoted to the main theorem. In Sect. \ref{s_compl} we discuss algorithmic complexity of the normalization procedure. In Sect. \ref{s_concl} we conclude.
\section{Related Work}\label{RW}
There are several works devoted to the Greibach normal form for graph grammars. The paper of Joost Engelfriet \cite{Engelfriet92} establishes that HRGs that produce \emph{languages of bounded degree} are equivalent to \emph{apex HRGs}; this result generalizes \emph{the double Greibach normal form} for CFGs. To recall, a language $L$ is of bounded degree if for some $M$ all nodes of all hypergraphs in $L$ have degree not exceeding $M$. Obviously, there are substantive examples of languages of unbounded degree: for instance, the language of star graphs, which have one node and arbitrarily many edges outgoing from it. Besides, the property of being apex is stronger than the weak Greibach normal form we are interested in.

In the paper of Christina Jansen et al. \cite{Jansen11} \emph{the local Greibach normal form} is presented. The authors prove that \emph{data structure grammars} (it is a specific kind of grammars that generate so-called heap configurations) can be transformed into grammars in the local Greibach normal form; after the proof the authors point out that the normalization can be generalized to HRGs of bounded degree. However, the authors also note that their procedure being algorithmically efficient cannot be generalized to all HRGs.

The weak Greibach normal form in the sense we are interested in is introduced for another type of graph grammars. Namely, in \cite{Engelfriet90.2} it is proved that each B-eNCE grammar (here we do not consider definitions regarding this formalism) is equivalent to a B-eNCE grammar in the weak Greibach normal form. In \cite{Engelfriet90} it is shown that B-edNCE grammars and HRGs have the same recognizing power in some sense (namely, since B-edNCE grammars produce usual graphs with labeled nodes and edges while HRGs produce hypergraphs with labeled edges, the result is established w.r.t. two translations, from graphs to hypergraphs and vice versa). However, it seems to be impossible that these two results can be combined in order to obtain the normal form for HRGs.

Therefore, to our best knowledge, the question of whether each HRG is equivalent to an HRG in the weak Greibach normal form has hitherto remained open. We answer it in this paper.
\section{Preliminaries}\label{s_defn}
This section is concerned with definitions related to hypergraphs. All of them are taken from \cite{Drewes97}. Note that we use a slightly different notation from that in the handbook mentioned above.
\subsection{Hypergraphs, Sub-hypergraphs}\label{hyp_def}

$\mathbb{N}$ includes $0$. 	The set of integers from $1$ to $n$ is denoted by $[1,n]$. 

It is convenient for us to use the following notation: if $\{i_1,\dots,i_k\}$ is an indexed set of integers such that for $m<n\;i_m<i_n$ holds, then it is called index-ordered and the set is denoted as $\{i_1,\dots,i_k\}_{IO}$. 
\\
The set $\Sigma^*$ is the set of all strings over the alphabet $\Sigma$ including the empty string $\varepsilon$. The length $|w|$ of the word $w$ is the number of positions in $w$. The $i$-th symbol in $w$ is denote by $w(i)$ ($1\le i\le |w|$). $\Sigma^+$ denotes the set of all nonempty strings. The set $\Sigma^\circledast$ is the set of all strings consisting of distinct symbols. The set of all symbols contained in the word $w$ is denoted by $[w]$. If $f:\Sigma\to\Delta$ is a function from one set to another, then it is naturally extended as a function $f:\Sigma^*\to\Delta^*$ ($f(\sigma_1\dots\sigma_k)=f(\sigma_1)\dots f(\sigma_k)$).

Let $C$ be some fixed set of labels, for which the function $type: C\to \mathbb{N}$ is considered.
\begin{definition}\label{hypergraph}
	\emph{A hypergraph over $C$} is a tuple $G=\langle V, E, att, lab, ext \rangle$ where $V$ is the set of \emph{nodes}, $E$ is the set of \emph{hyperedges}, $att: E\to V^\circledast$ assigns an ordered set of \emph{attachment} nodes to each edge, $lab: E \to C$ labels each hyperedge by some element of $C$ in such a way that $type(lab(e))=|att(e)|$ whenever $e\in E$, and $ext\in V^\circledast$ is an ordered set of \emph{external} nodes. 
	
	Components of a hypergraph $G$ are denoted by $V_G, E_G, att_G, lab_G, ext_G$.
\end{definition}
In the remainder of the paper, hypergraphs are simply called graphs, and hyperedges are simply called edges. The set of all graphs with labels from $C$ is denoted by $\mathcal{H}(C)$. 
In this work, figures contain graph drawings and sketches. In drawings, nodes are depicted by black dots, edges are denotes as labeled boxes, $att$ is represented with numbered lines (called ``tentacles''), external nodes are depicted by numbers in brackets. If an edge has type 2, it is depicted by an arrow. If we are not interested in a specific form of a graph but we want to have a closer look at some its part, we depict the whole graph as an area but draw its part of interest in detail.
\begin{example}
	This is a graph:
	
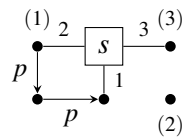
\begin{center}
		{\tikz[baseline=.1ex]{
				\node[] (R) {};
				\node[node,above right=2mm and 0mm of R,label=above:{\scriptsize $(1)$}] (N1) {};
				\node[node,below=5.5mm of N1] (N2) {};
				\node[hyperedge,right=5.5mm of N1] (E) {$s$};
				\node[node,below=3.7mm of E] (N3) {};
				\node[node,right=5.5mm of E,label=above:{\scriptsize $(3)$}] (N4) {};
				\node[node,below=5.5mm of N4,label=below:{\scriptsize $(2)$}] (N5) {};

				\draw[->,black] (N1) -- node[left] {\small $p$} (N2);
				\draw[->,black] (N2) -- node[below] {\small $p$} (N3);
				\draw[-,black] (N3) -- node[right] {\scriptsize 1} (E);
				\draw[-,black] (N1) -- node[above] {\scriptsize 2} (E);
				\draw[-,black] (N4) -- node[above] {\scriptsize 3} (E);
		}}
	\end{center}
\end{example}
If $G$ is a graph, and $e\in E_G$ is labeled by $a$, then $G$ can be denoted by $G(e:a)$.
\begin{definition}\label{type}
	The function $type$ (or $type_G$ to be exact) returns the number of nodes attached to some edge in a graph $G$: $type_G(e):=|att_G(e)|$.
	If $G$ is a graph, then $type(G):=|ext_G|$.
\end{definition}
\begin{definition}
	A sub-hypergraph (or just subgraph) $H$ of a hypergraph $G$ is a hypergraph such that $V_H\subseteq V_G$, $E_H\subseteq E_G$, and for all $e\in E_H$ $att_H(e)=att_G(e)$, $lab_H(e)=lab_G(e)$.
\end{definition}
%Note that each subgraph $H$ of the graph $G$ can be uniquely defined by $E_H$ and $ext_H$ (particularly, $V_H$ equals the set of all the attachment nodes of edges from $E_H$). Let us then denote $H$ by $\langle E_H, ext_H \rangle_G$ or even $\langle E_H \rangle_G$ if the specific form of $ext_H$ is not important.
\begin{definition}
	If $H=\langle \{v_i\}_{i=1}^n,\{e_0\},att,lab,v_1\dots v_n\rangle$, $att(e_0)=v_1\dots v_n$ and $lab(e_0)=a$, then $H$ is called \emph{a handle}. It is denoted by $\circledcirc(a)$.
\end{definition}
\begin{definition}
	\emph{An isomorphism} between graphs $G$ and $H$ is a pair of bijective functions $\mathcal{E}: E_G\to E_H$, $\mathcal{V}: V_G\to V_H$ such that $att_H\circ\mathcal{E}=\mathcal{V}\circ att_G$, $lab_G=lab_H\circ\mathcal{E}$, $\mathcal{V}(ext_G)=ext_H$. In this work, we do not distinguish between isomorphic graphs.
\end{definition}

\subsection{Replacement}	
This procedure is defined in \cite{Drewes97}. In short, the replacement of an edge $e_0$ in $G$ with a graph $H$ can be done if $type(e_0)=type(H)$ as follows:
\begin{enumerate}
	\item Remove $e_0$;
	\item Insert an isomorphic copy of $H$ (namely, $H$ and $G$ have to consist of disjoint sets of nodes and edges);
	\item For each $i$, fuse the $i$-th external node of $H$ with the $i$-th attachement node of $e_0$.
\end{enumerate}
To be more precise, the set of edges in the resulting graph is $(E_G\setminus\{e_0\})\cup E_H$, and the set of nodes is $V_G\cup (V_H\setminus ext_H)$. The result is denoted by $G[H/e_0]$.

\subsection{Hyperedge Replacement Grammars}
\begin{definition}
	A \emph{hyperedge replacement grammar (HRG)} is a tuple $HGr=\langle N, \Sigma, P, S\rangle$, where $N$ is a finite alphabet of nonterminal symbols, $\Sigma$ is a finite alphabet of terminal symbols ($N\cap\Sigma=\emptyset$), $P$ is a set of productions, and $S\in N$. Each production is of the form $A\to H$ where $A\in N$, $H\in\mathcal{H}(N\cup \Sigma)$ and $type(A)=type(H)$. For $\pi=A\to H$ we denote $lhs(\pi)=A,rhs(\pi)=H$.
\end{definition}

Edges labeled by terminal (nonterminal) symbols are called \emph{terminal edges} (\emph{nonterminal edges} resp.). If a graph contains terminal edges only, it is called \emph{terminal}.

If $G$ is a graph, $e_0\in E_G$, $lab(e_0)=A$ and $\pi=A\to H\in P$, then $G$ directly derives $G[H/e_0]$ (we write $G\Rightarrow G[H/e_0]$ or $G\underset{\pi}{\Rightarrow} G[H/e_0]$). The transitive reflexive closure of $\Rightarrow$ is denoted by $\overset{\ast}{\Rightarrow}$. If $G \overset{\ast}{\Rightarrow} H$, then $G$ is said to derive $H$. The corresponding sequence of production applications is called a derivation.
\begin{definition}
	The \emph{language generated by an HRG $HGr=\langle N, \Sigma, P, S\rangle$} is the set of graphs $H\in \mathcal{H}(\Sigma)$ such that $\circledcirc(S)\overset{\ast}{\Rightarrow} H$. It is denoted by $L(HGr)$.
	
	A graph language $L$ is said to be \emph{a hypergraph context-free language (HCFL)} if it is generated by some HRG.
	
	Two grammars are said to be equivalent if they generate the same language.
\end{definition}
Further we simply write $A \overset{\ast}{\Rightarrow} G$ instead of $\circledcirc(A) \overset{\ast}{\Rightarrow} G$.

\section{Weak Greibach Normal Form}\label{s_GNF}
In this section we introduce the formal definition of the normal form we are interested in, and establish a simple property of languages generated by grammars in this normal form.
\begin{definition}
	An HRG $HGr$ is in \emph{the weak Greibach normal form (WGNF)} if there is exactly one terminal edge in the right-hand side of each production. Formally, $\forall (X\to H)\in P_{HGr}$ $ \exists!e_0\in E_H:lab_H(e_0)\in \Sigma_{HGr}$.
\end{definition}
\begin{example}
	The grammar $HGr_1=\langle\{S\},\{a\},P_1,S\rangle$ is in the WGNF where $P_1$ contains the following rules ($type(S)=type(a)=1$):
	
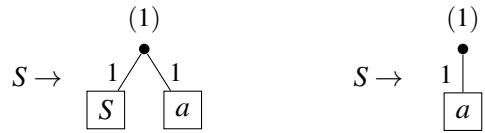
\begin{center}
		\begin{tikzpicture}
		\node (S) {$S\to$};
		\node[node, above right=0.5mm and 9mm of S,label=above:{\small $(1)$}] (v0) {};
		\node[hyperedge,below left=5mm and 2mm of v0] (e1) {$S$};
		\draw[-,black] (v0) -- node[left] {\small 1} (e1);
		\node[hyperedge,below right=5mm and 2mm of v0] (e2) {$a$};
		\draw[-,black] (v0) -- node[right] {\small 1} (e2);
		\node[right=36mm of S] (S2) {$S\to$};
		\node[node, above right=0.5mm and 6mm of S2,label=above:{\small $(1)$}] (v01) {};
		\node[hyperedge,below=5mm of v01] (e11) {$a$};
		\draw[-,black] (v01) -- node[left] {\small 1} (e11);
		\end{tikzpicture}
	\end{center}
	This grammar generates the language of star 1-edged $a$-labeled hypergraphs, which is unbounded.
\end{example}
\begin{remark}\label{rem_terminal_edges}
	If $\langle N,\Sigma,P,S\rangle$ is in the WGNF and $X\overset{k}{\Rightarrow} G$ for some $X\in N$, then $G$ has exactly $k$ terminal edges (each production adds exactly one terminal edge).
\end{remark}
Note that not each hypergraph context-free language can be generated by some HRG in the WGNF. This follows from 
\begin{example}\label{ex_nonWGNF}
	Consider the HRG $HGr_2=\langle\{T\},\{b\},P_2,T\rangle$ where $P_2$ contains the following rules \\($type(T)=type(b)=0$):
	\begin{center}
		\begin{tikzpicture}
		\node (T) {$T\to$};
		\node[hyperedge,above right=-3mm and 2mm of T] (e1) {$T$};
		\node[node, below=3mm of e1] (v0) {};
		\node[right=30mm of T] (T2) {$T\to$};
		\node[hyperedge,right=2mm of T2] (e2) {$b$};
		\end{tikzpicture}
	\end{center}
	The first production just adds an isolated node; thus this grammar produces graphs that have exactly one edge labeled by $b$ and arbitrarily many isolated nodes.
	If there is an equivalent $HGr^\prime=\langle N,\{b\},P^\prime,S^\prime\rangle$ in the WGNF, then each right-hand side of each production in $P^\prime$ contains exactly one terminal edge. Note that if $S^\prime\overset{k}{\Rightarrow} G,G\in\mathcal{H}(\{b\})$ for some $k$ in $HGr^\prime$, then $G$ has $k$ terminal edges (Remark \ref{rem_terminal_edges}); hence $k$ has to be equal to $1$ and therefore $S^\prime\to G\in P^\prime$. However, there are infinitely many graphs in $L(HGr)$ while $|P^\prime|<\infty$.
	\qed
\end{example}

It is seen that the reason of this problem is in isolated nodes. In the string case if a language contains the empty word, then it cannot be generated by a grammar in the WGNF due to obvious reasons: each production adds at least one terminal symbol so it is impossible to produce the empty word. In the case of graphs, we also have to prohibit somehow undesired languages with ``too many'' isolated nodes occuring in graphs.

Below we describe the precise class of languages generated by grammars in the weak Greibach normal form.
We denote by $esize(G)$ the number of edges in $G$, and by $isize(G)$ the number of isolated nodes in $G$.
\begin{definition}
	An HCFL $L$ is said to be \emph{isolated-node bounded} (ibHCFL) if there is a constant $M$ such that, for each $G\in L$, $isize(G)< M*esize(G)$. An HRG $HGr$ is called isolated-node bounded if it generates an ibHCFL.
\end{definition}
\begin{theorem}\label{Theorem}
	Each HRG $HGr=\langle N,\Sigma,P,S\rangle$ in the WGNF generates an ibHCFL.
\end{theorem}
\begin{proof}
	Let $M=\max_{A\to H\in P}\{isize(H)\}+1$. Define $esize_T(G)$ as the number of terminal edges in $G$. We prove by induction on $k$ that for each graph $G$ such that $X\overset{k}{\Rightarrow} G, X\in N$ $isize(G)< M*esize_T(G)$. Remark \ref{rem_terminal_edges} implies that $esize_T(G)=k$.	
	
	\textbf{Basis.} If $k=0$, then $G=\circledcirc(X)$, and the statement is trivial ($0<M$).
	
	\textbf{Step.} Let $X\overset{k-1}{\Rightarrow} H \Rightarrow G$. By the induction hypothesis, $isize(H)< M*esize_T(H)$. The number of isolated nodes appeared at the last step does not exceed $M$ by the definition of $M$. Then $isize(G)< M*esize_T(H)+M=M*esize_T(G)$.
	
	Note finally that $esize_T(G)\le esize(G)$.
\end{proof}

The other direction of this statement is of central interest in this paper. The next section is devoted to it.

\section{Transformation of HRGs Generating IBHCFLs into HRGs in the Weak Greibach Normal Form}\label{s_proof}
The main theorem we are going to prove is the following:
\begin{theorem}\label{THEOREM}
	Each ibHCFL $L$ can be generated by an HRG in the WGNF.
\end{theorem}
In other words, this theorem states that each context-free hypergraph language satisfying the isolated-node boundedness property can be generated by a grammar in the weak Greibach normal form. 

Structurally, the proof of this theorem is based on the corresponding one for string context-free grammars. Let us recall the mains steps of the latter in brief (see details in \cite{Aho72}). 

\textit{The input of the algorithm is a context-free grammar $CFG$ \emph{that does not generate the empty word} (this property is related to isolated-node boundedness). The desired output is an equivalent grammar in the weak Greibach normal form (i.e.~in which each production is of the form $A\to aA_1\dots A_n$ where $a$ is terminal while $A_1,\dots,A_n$ are nonterminal).}
\begin{enumerate}
	\item Useless rules and symbols, $\varepsilon$-rules (i.e.~rules of the form $A\to\varepsilon$) and chain rules (i.e.~rules of the form $A\to B$ for $A,B$ being nonterminal) are eliminated.
	\item\label{recprodstr} It is shown how to eliminate recursive $A$-productions, i.e.~productions of the form $A\to A\alpha$ for some fixed nonterminal symbol $A$. The trick is to move $A$ from the left side of $A\alpha$ to the right side. It is done as follows: if $A\to A\alpha_1|\dots|A\alpha_m|\beta_1|\dots|\beta_p$ are all the $A$-productions (here $|$ means enumeration of productions), then they are replaced by the productions $A\to \beta_1 A^\prime|\dots|\beta_pA^\prime|\beta_1|\dots|\beta_p$ and $A^\prime\to\alpha_1|\dots|\alpha_m|\alpha_1 A^\prime|\dots|\alpha_mA^\prime$.
	\item Left recursion is completely eliminated: nonterminals are numbered as $A_1,\dots,A_n$ and then a procedure involving application of step \ref{recprodstr} is done such that its result is a grammar where each production is either of the form $A_i\to A_j\alpha$ for $i<j$ or of the form $A_i\to b\alpha$ for $b$ being terminal.
	\item By taking compositions of the above productions and adding new nonterminal symbols the grammar is normalized. The resulting grammar includes rules of the form $A\to b\alpha$ ($b$ is terminal, $\alpha$ is a string of nonterminals) only.
\end{enumerate}

Our goal is to recreate this proof for HRGs. Note that steps 2-4 of the above plan actively exploit string nature of CFGs: transformations of productions are based on movements of symbols from the leftmost position to the rightmost one. Graphs in general do not have leftmost positions unlike strings, so we are going to distinguish arbitrary edges in graphs so they play the role of ``the leftmost symbol''. It is done by means of $\delta$ (see the proof below).

In the proof we use the following
\begin{lemma}\label{lemma_repl}
		If $\pi=A\to G(e_0:B)$ is a production of a grammar $HGr=\langle N,\Sigma,P,S\rangle$ for $A,B$ being nonterminal, and $B\to H_1,\dots,B\to H_k$ are all the productions in $HGr$ with $B$ in the left-hand side, then replacing $\pi$ by productions $A\to G[H_1/e_0],\dots,A\to G[H_k/e_0]$ does not change the language generated.
\end{lemma}
It directly follows from the well-known context-freeness lemma.
\subsection{Eliminating Useless, Edgeless and Chain Productions}
Our first objective is to eliminate useless productions, productions with no edges in the right-hand side and productions with one nonterminal edge only in the right-hand side. It appears that only this step requires isolated-node boundedness. In the below transformations we provide theoretical reasonings that they can be done irregarding their algorithmic realization; however, the latter can be done similarly to the string case (see \cite{Aho72}).
\begin{definition}
	A nonterminal symbol $A$ in a grammar $HGr$ is \emph{useless} if there is no derivation of the form $S\overset{\ast}{\Rightarrow} H(e:A)\overset{\ast}{\Rightarrow} G$ in this grammar where $G$ is terminal.
\end{definition}
\begin{transform}[eliminating useless symbols]\label{useless}
	\leavevmode
	\\
	\textbf{Input:} an HRG $HGr$.
	\\
	\textbf{Output:} an equivalent HRG $HGr^\prime$ without useless nonterminal symbols.
\end{transform}
\noindent
It suffices to remove all the useless symbols and all the productions containing them. This does not affect the language generated.
\begin{definition}
	A graph is called \emph{edgeless} if it does not contain edges. A production $A\to G$ is called \emph{edgeless} if $G$ is edgeless.
\end{definition}
\begin{transform}[eliminating edgeless productions]\label{edgeless}
	\leavevmode
	\\
	\textbf{Input:} an HRG $HGr$ that generates an ibHCFL without useless symbols.
	\\
	\textbf{Output:} an HRG $HGr^\prime$ without edgeless productions such that $L(HGr^\prime)=L(HGr)$.
\end{transform}
\noindent
\textbf{Method.}
\\
	Let $HGr=\langle N,\Sigma,P,S\rangle$. Let $Null=\{(A; H)\mid E_H=\emptyset,\:A\in N,\:A\overset{\ast}{\Rightarrow}H\}$. This set is finite, because otherwise there is a symbol $A_0$ for which arbitrarily large edgeless graphs $H$ exist such that $(A_0; H)\in Null$; this contradicts the fact that $L(HGr)\in\mathrm{ibHCFL}$ (it is important that $A_0$ is not useless). 
	
	Let $B\to G\in P$ and $E_G=\{e_1,\dots,e_n\}$. Let $P_1$ contain the rules of the form $B\to G[H_1/e_1]\dots[H_n/e_n]$ where $H_i=\circledcirc(lab(e_i))$ for at least one $i$ (then replacement of $e_i$ by $H_i$ changes nothing) and $(lab(e_i);H_i)$ belongs to $Null$ otherwise. It is argued that $HGr^\prime=\langle N,\Sigma,P_1,S\rangle$ does not have productions with edgeless right-hand sides (this follows from the construction) and that $L(HGr^\prime)=L(HGr)$.
\qed

\begin{definition}
	A production $A\to G$ is called a chain production if $E_G=\{e\}$ and $lab_G(e)$ is nonterminal.
\end{definition}
\begin{example}
	The first production in the grammar $HGr_2$ is chain.
\end{example}
\begin{transform}[eliminating chain productions]\label{1edged}
	\leavevmode
	\\
	\textbf{Input:} an HRG $HGr$ that generates an ibHCFL without useless symbols.
	\\
	\textbf{Output:} an HRG $HGr^\prime$ without chain productions such that $L(HGr^\prime)=L(HGr)$.
\end{transform}
	\noindent
	\textbf{Method.}
	\\
	Let $HGr=\langle N,\Sigma,P,S\rangle$. Consider the set $Chain=\{(A;H)\mid E_H=\{e_0\},lab_H(e_0)=B,\:A,B\in N,\:A\overset{\ast}{\Rightarrow}H\}$. Note that $Chain$ is finite (again, otherwise one can derive a graph with arbitrarily many isolated nodes having a fixed number of edges). Let $HGr^\prime=\langle N,\Sigma,P^{\prime\prime},S\rangle$ where $P^\prime=P\setminus \{A\to H\mid (A;H)\in Chain\}$ and $P^{\prime\prime}=P^\prime\cup\{A\to G\mid \exists H:\:(A;H)\in Chain,H\to G\in P^\prime\}$. Thus we removed all the productions having a graph with one edge in the right-hand side. It can be easily shown that $L(HGr^\prime)=L(HGr)$.
\qed
\begin{remark}
	If $HGr$ in Transformation \ref{1edged} does not have edgeless productions, then $HGr^\prime$ does not have them either. Transformations \ref{edgeless} and \ref{1edged} applied to a grammar without useless symbols transform them into a grammar without useless symbols too.
\end{remark}
\begin{example}
	The grammar $HGr_2$ from Example \ref{ex_nonWGNF} cannot be turned into an equivalent one without edgeless and chain productions.
\end{example}
The above procedures complete the first step of normalization. 

\subsection{Defining and Eliminating Recursive Productions}
Now we start proving Theorem \ref{THEOREM}. Let $HGr=\langle N,\Sigma,P,S\rangle$ be a grammar that generates $L$. Applying Transformations \ref{useless}, \ref{edgeless}, \ref{1edged}, we can assume that $HGr$ does not have useless symbols, edgeless productions and chain productions. Thus, the right-hand side of each production has at least two edges or one terminal edge.
	
Let us carefully examine what happens in the string case at Step \ref{recprodstr}. Consider the following
\begin{example}\label{rederstr}
	Let $A\to Ac|Ad|Be$ be all the $A$-productions in a string context-free grammar. After Step \ref{recprodstr} there are the following ones: $A\to BeA^\prime|Be,\; A^\prime\to c|d|cA^\prime|dA^\prime$. Below we show how a derivation in the old grammar is remodeled in the new one:
	
	\begin{tabular}{ll}
		\textbf{Old:}& $A\Rightarrow Ac\Rightarrow Adc\Rightarrow Addc\Rightarrow Beddc$.
		\\
		\textbf{New:}& $A\Rightarrow BeA^\prime\Rightarrow BedA^\prime\Rightarrow BeddA^\prime\Rightarrow Beddc$.\\
	\end{tabular}
	\\
	The underlying idea is to invert the derivation: in the new derivation one starts with $Be$ and applies new productions in reverse order w.r.t. to the old derivation.
\end{example}
	The same idea is used in the graph case. The difficulty is that we do not have the leftmost symbol in productions. However, it suffices to distinguish an arbitrary edge in the right-hand side of each production to play the role of the leftmost symbol. We do this by means of the function $\delta$. 
	\begin{definition}
	Let us fix an arbitrary function $\delta$ which acts on $P$ in such a way that $\delta(A\to H)$ belongs to $E_H$. We denote $\mu(\pi)=lab(\delta(\pi))$. 
	\end{definition}
	Now $\mu(\pi)$ plays the role of the ``leftmost'' symbol of a production $\pi$. Then we define recursive productions as expected.
	\begin{definition}
		A production $\pi$ is recursive if $lhs(\pi)=\mu(\pi)$. A production $\pi$ is an $A$-production if $lhs(\pi)=A$.
	\end{definition}
	\begin{definition}
		A derivation $A\Rightarrow H_1\Rightarrow\dots\Rightarrow H_k$ is called a $\delta$-derivation if each of its productions is applied to the edge that is $\delta$ of the previous production. Formally, if $A\underset{\pi_1}{\Rightarrow} H_1\underset{\pi_2}{\Rightarrow} \dots \underset{\pi_k}{\Rightarrow} H_k$, then $\pi_i$ has to be appied to $\delta(\pi_{i-1})$ (in the subgraph that appears after the $(i-1)$-th step). We say that the final label of such a derivation is $\mu(\pi_k)$. Note that $\mu(\pi_i)$ has to be nonterminal for $i<k$ (but $\mu(\pi_k)$ can be terminal).
	\end{definition}
	Below we study how to eliminate all the recursive $A$-productions for some fixed $A\in N$. In the string case the idea is quite simple: one moves $A$ from the beginning of the string to its end in order to invert a derivation process. Here we also exploit the idea of inverting the derivation but, of course, we have to take into account more complex and general graph structures.
	
	Let $A\underset{\pi_1}{\Rightarrow} H_1\underset{\pi_2}{\Rightarrow} \dots \underset{\pi_k}{\Rightarrow} H_k$ be a $\delta$-derivation such that $lhs(\pi_i)=A$ for $i=1,\dots,k$ and $\mu(\pi_k)\ne A$ (compare this with the old derivation in Example \ref{rederstr}). We provide the following intuition behind the below construction. Imagine that one has a hole in his sweater (this metaphor is related to that in \cite{Engelfriet92}), and he sews it up starting from edges (i.e.~he applies $\pi_1,\dots,\pi_{k-1}$) and finishing by sewing on a patch in the center of the hole (i.e.~he applies $\pi_k$). In the new derivation one starts with the patch (that is, the right-hand side of $\pi_k$), then he sews up right-hand sides of $\pi_{k-1},\dots,\pi_2$ in this order (thus the patch grows and becomes bigger) and finishing by connecting it with edges of the hole w.r.t. the production $\pi_1$ (see also Fig. \ref{fig_rederivation} which provides a sketch of this process).
	
	The major problem of this idea is the following: if we want to invert a derivation such as on Figure \ref{fig_source_der} and to obtain a new one, such as one drawn on Figure \ref{fig_new_der}, we have to carefully control external nodes. Namely, in the old derivation external nodes of $H$ are predefined by $R_1$ while in the new derivation $R_1$ appears at the very last step. In order to place external nodes correctly, we introduce complex nonterminal symbols that describe how external nodes of a graph on the current derivation step correspond to external nodes of a graph on the next and on the last step. This description is done by means of partial functions.
	
	Now we proceed with formal realization of this idea. Let $t=type(A)$. Let $\rho_1=A\to R_1,\dots,\rho_K=A\to R_K$ be all the recursive $A$-productions in $P$ and let $\gamma_1=A\to G_1,\dots, \gamma_L=A\to G_L$ be the remaining $A$-productions in $P$. Our goal is to construct a grammar equivalent to $HGr$ without recursive $A$-productions.
	Firstly, we simply remove them: $P_1:=P\setminus \{\rho_1,\dots,\rho_K\}$. In order to compensate for the lack of these rules we add new ones accordingly to the procedure below. 
	\begin{definition}
		$f:X\to Y$ is a partial function if $f$ is a function on some subset $X^\prime$ of $X$. We denote the domain $X^\prime$ of $f$ by $Dom(f)$, and the range of $f$ by $Ran(f)$.
	\end{definition}
	\begin{definition}
		$f:X\to Y$ is a partial bijection if $f$ is a partial function such that $f|_{Dom(f)}$ is a bijection.
	\end{definition}
	\begin{definition}
		If $f:X\to Y$, $g:Y\to Z$ are partial functions, then $g\circ f:X\to Z$ is a partial function defined on $Dom(f)\cap f^{-1}(Dom(g))$ that acts on this set as a usual composition.
	\end{definition}
	
	\begin{figure}
		\begin{subfigure}{0.5\linewidth}
		\centering	
		\begin{tikzpicture}
		\renewcommand{\scale}{1.7}
		\node (v0) at (0.5*\scale,0.5*\scale) {$A\to$};
		\node (v1) at (1.2*\scale,0.866*\scale) {};
		\begin{scope}[fill opacity=1]
		\filldraw[pattern=north west lines,pattern color=blue] ($(v1)+(-0.2*\scale,0.2*\scale)$) 
		to[out=45,in=180] ($(v1) + (0*\scale,0.25*\scale)$) 
		to[out=0,in=180] ($(v1) + (0.5*\scale,0.15*\scale)$)
		to[out=0,in=180] ($(v1) + (1*\scale,0.25*\scale)$)
		to[out=0,in=135] ($(v1)+(1.2*\scale,0.2*\scale)$)
		to[out=-45,in=90] ($(v1)+(0.8*\scale,-0.866*\scale)$)
		to[out=270,in=0] ($(v1)+(0.5*\scale,-1.116*\scale)$)
		to[out=180,in=-60] ($(v1)+(0.1*\scale,-0.634*\scale)$)
		to[out=120,in=270] ($(v1)+(-0.3*\scale,0*\scale)$)
		to[out=90,in=-135] ($(v1)+(-0.2*\scale,0.2*\scale)$)
		;
		\end{scope}		
		\filldraw[color=white,text=black] ($(v1)+(0.5*\scale,-0.334*\scale)$) circle [radius=0.2*\scale] node {\Large $G$};
		\node (e1) at ($(v1)+(0.9*\scale,0.4*\scale)$) {\scriptsize $(1)$};
		\node (e2) at ($(v1)+(-0.2*\scale,0.35*\scale)$) {\scriptsize $(t)$};
		\draw[-,black,opacity=0]  (e1) to [bend left=10]
		node[pos=0.55,sloped,above,black,opacity=1] {\dots}
		($(v1)+(1.3*\scale,0.1*\scale)$);
		\draw[-,black,opacity=0]  (e2) to [bend right=5]
		node[pos=0.55,sloped,above,black,opacity=1] {\dots}
		($(v1)+(-0.35*\scale,0.1*\scale)$);
		\end{tikzpicture}
	\caption{The production $\gamma$.}
	\end{subfigure}
	\begin{subfigure}{0.5\linewidth}
			\centering	
				\begin{tikzpicture}
				\renewcommand{\scale}{1.7}
				\node (v0) at (0.5*\scale,0.5*\scale) {$A\to$};
				\node (v1) at (2.1*\scale,1*\scale) {};
				\begin{scope}[fill opacity=1]
				\filldraw[pattern=north east lines,pattern color=cyan] ($(v1)+(0*\scale,0.3*\scale)$) 
				to[out=0,in=150] ($(v1) + (0.9*\scale,-0.3*\scale)$) 
				to[out=-30,in=90] ($(v1) + (1.2*\scale,-0.7*\scale)$)
				to[out=-90,in=0] ($(v1) + (0.75*\scale,-1.2*\scale)$)
				to[out=180,in=0] ($(v1) + (0*\scale,-0.95*\scale)$)
				to[out=180,in=0] ($(v1) + (-0.75*\scale,-1.2*\scale)$)
				to[out=180,in=-90] ($(v1) + (-1.1*\scale,-0.7*\scale)$)
				to[out=90,in=180] ($(v1)+(0*\scale,0.3*\scale)$) 
				%%to[out=,in=] ($(v1) + (*\scale,*\scale)$)
				;
				\end{scope}
				\filldraw[color=white] ($(v1)+(-0.3*\scale,-0.6*\scale)$) rectangle ($(v1)+(0.2*\scale,-0.1*\scale)$);
				\filldraw[pattern=horizontal lines,pattern color=red] ($(v1)+(-0.3*\scale,-0.6*\scale)$) rectangle ($(v1)+(0.2*\scale,-0.1*\scale)$);
				\node (R) at ($(v1)+(0*\scale,-1.2*\scale)$) {\large $R$};
				\filldraw[color=white,text=black] ($(v1)+(-0.05*\scale,-0.35*\scale)$) circle [radius=0.12*\scale] node {$A$};
				
				\node (e1) at ($(v1)+(0.77*\scale,0*\scale)$) {\scriptsize $(1)$};
				\draw[-,black,opacity=0]  (e1) to [bend left=5] node[pos=0.55,sloped,above,black,opacity=1] {\dots} ($(v1)+(1.2*\scale,-0.5*\scale)$);
				\node (e2) at ($(v1)+(-1*\scale,0*\scale)$) {\scriptsize $(t)$};
				\draw[-,black,opacity=0]  (e2) to [bend right=5]
				node[pos=0.55,sloped,above,black,opacity=1] {\dots}
				($(v1)+(-1.1*\scale,-0.5*\scale)$);
				
				\draw[-,black] ($(v1)+(0.2*\scale,-0.25*\scale)$) to node[above] {\scriptsize 1} ($(v1)+(0.5*\scale,-0.2*\scale)$);
				\filldraw[color=black] ($(v1)+(0.5*\scale,-0.2*\scale)$) circle [radius=0.02*\scale] node {};
				
				\draw[-,black] ($(v1)+(0.2*\scale,-0.5*\scale)$) to node[above] {\scriptsize 2} ($(v1)+(0.5*\scale,-0.55*\scale)$);
				\filldraw[color=black] ($(v1)+(0.5*\scale,-0.55*\scale)$) circle [radius=0.02*\scale] node {};

				\draw[-,black] ($(v1)+(-0.3*\scale,-0.25*\scale)$) to node[below] {\scriptsize $t$} ($(v1)+(-0.6*\scale,-0.2*\scale)$);
				\filldraw[color=black] ($(v1)+(-0.6*\scale,-0.2*\scale)$) circle [radius=0.02*\scale] node {};
				
				\draw[-,black,opacity=0]  ($(v1)+(0.5*\scale,-0.55*\scale)$) to [bend left=25]
				node[pos=0.55,sloped,above,black,opacity=1] {\dots}
				($(v1)+(-0.3*\scale,-0.9*\scale)$);
				\end{tikzpicture}
			\caption{The production $\rho$.}
		\end{subfigure}
	\caption{The productions $\gamma$ and $\rho$.}\label{GR}
	\end{figure}
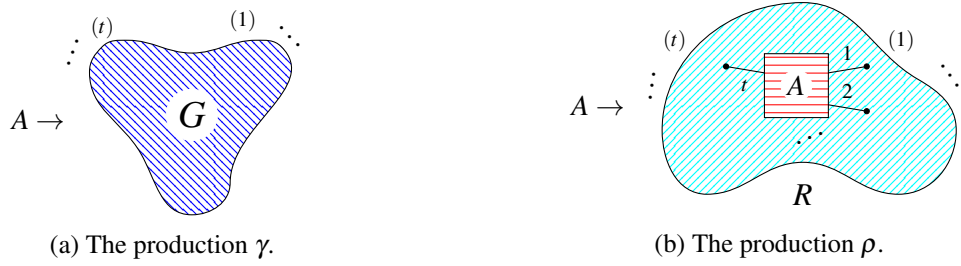
	Let $\gamma=\gamma_i$ for some $i$ and let $\rho = \rho_j$ for some $j$. We define $e_0:=\delta(\rho)$, $\widetilde{e}:=\delta(\gamma)$, $G:=rhs(\gamma)$ and $R:=rhs(\rho)$. Illustrations of these productions are presented in Fig. \ref{GR}. Firstly, we extend $N$ by \emph{new} nonterminal symbols of the form $(A,f,g)$ where $f,g:[1,t]\to[1,t]$ are two partial bijections; the resulting set is denoted by $N^\prime$. Then we add to $P_1$ rules of the following three forms:
	\\
	\textbf{Type I. }
	Recall the metaphor about a sweater and a patch. If one imagines an $A$-labeled edge at the beginning of a derivation as a hole in a sweater, then productions of type I are designed to add a patch $G$ at the very beginning of sewing up (= of a derivation). Partial functions are used to describe what external nodes of $G$ (edges of a patch) are equal to attachment nodes of the $A$-labelled edge (edges of a hole).
	\\
	Let $f,g:[1,t]\to[1,t]$ be two partial bijections. Informally, $f$ specifies what nodes of $ext_G$ are external on the second step of the inverted derivation (see again Fig. \ref{fig_new_der}); $g$ specifies what external nodes of the second step graph are external at the last step.
	Let $Dom(g\circ f)=\{p_1,\dots,p_k\}_{IO}$.
	We set
	\begin{itemize}
	\item $V^\prime:=V_G\cup\{u_1,\dots,u_{t-k}\}$ for $u_1,\dots, u_{t-k}$ being new nodes;
	\item $E^\prime :=E_G\cup\{e^\prime\}$ where $e^\prime$ is new;
	\item For $e\in E_G$ we set $att^\prime(e):=att_G(e)$ and $lab^\prime(e)=lab_G(e)$;
	\item For $e^\prime$ we set $att^\prime(e^\prime):=ext_Gu_1\dots u_{t-k}$;
	\item $lab^\prime(e^\prime):=(A,f,g)$;
	\item Let $[1,t]\setminus Ran (g\circ f)=\{j_1,\dots,j_{t-k}\}_{IO}$ (note that $g\circ f$ is bijective so $Dom(g\circ f)$ and $Ran(g\circ f)$ are of the same size). Then we set $ext^\prime(i):=u_r$, if $i=j_r,\;r\in\{1,\dots,t-k\}$ or $ext^\prime(i):=ext_G(p_r)$ for $i=i_r=g(f(p_r)),\;r\in\{1,\dots,k\}$.
	\end{itemize}
	We are ready to announce that $G^\prime=\langle V^\prime, E^\prime, att^\prime,lab^\prime,ext^\prime\rangle$. Finally, we introduce the following rule: $${\nu_{1,\gamma,f,g}:=A\to G^\prime}.$$ We define $\delta(\nu_{1,\gamma,f,g})=\delta(\gamma)$ (note that $\gamma$ is nonrecursive so is $\nu_{1,\gamma,f,g}$). 	
	\begin{remark}
		Type of $(A,f,g)$ equals $type((A,f,g))=t+t-k=2\cdot type(A)-|Dom(g\circ f)|$.
	\end{remark}
	Figure \ref{fig_I} illustrates this type of productions.
	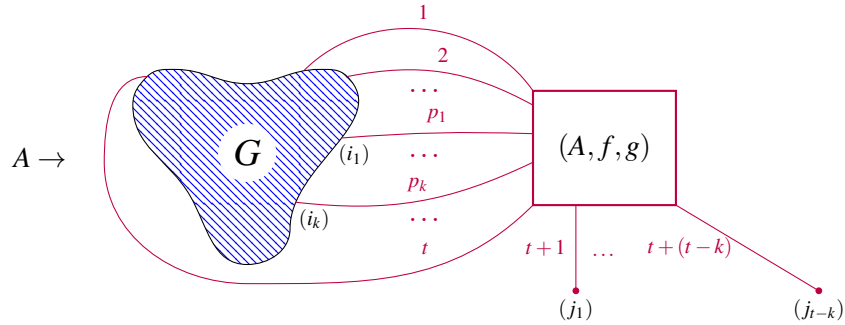
\begin{figure}
			\centering	
			\begin{tikzpicture}
			\renewcommand{\scale}{1.9}
			\node (v0) at (0.5*\scale,0.5*\scale) {$A\to$};
			\node (v1) at (1.45*\scale,0.866*\scale) {};
			\begin{scope}[fill opacity=1]
			\filldraw[pattern=north west lines,pattern color=blue] ($(v1)+(-0.2*\scale,0.2*\scale)$) 
			to[out=45,in=180] ($(v1) + (0*\scale,0.25*\scale)$) 
			to[out=0,in=180] ($(v1) + (0.5*\scale,0.15*\scale)$)
			to[out=0,in=180] ($(v1) + (1*\scale,0.25*\scale)$)
			to[out=0,in=135] ($(v1)+(1.2*\scale,0.2*\scale)$)
			to[out=-45,in=90] ($(v1)+(0.8*\scale,-0.866*\scale)$)
			to[out=270,in=0] ($(v1)+(0.5*\scale,-1.116*\scale)$)
			to[out=180,in=-60] ($(v1)+(0.1*\scale,-0.634*\scale)$)
			to[out=120,in=270] ($(v1)+(-0.3*\scale,0*\scale)$)
			to[out=90,in=-135] ($(v1)+(-0.2*\scale,0.2*\scale)$)
			;
			\end{scope}
			\filldraw[color=white,text=black] ($(v1)+(0.5*\scale,-0.334*\scale)$) circle [radius=0.2*\scale] node {\Large $G$};

			\node (Afg) at ($(v1)+(3*\scale,-0.3*\scale)$) { $(A,f,g)$};
			\draw[color=purple, text = black,thick] ($(Afg)+(-0.5*\scale,-0.4*\scale)$) rectangle ($(Afg)+(0.5*\scale,0.4*\scale)$);
			
			\draw[-,purple,opacity=1]  ($(Afg)+(-0.5*\scale,0.4*\scale)$) to [bend right=50] 
			node[above,purple,opacity=1] {\scriptsize 1}
			($(v1)+(0.9*\scale,0.245*\scale)$);
			
			\draw[-,purple,opacity=1]  ($(Afg)+(-0.5*\scale,0.3*\scale)$) to [bend right=20] 
			node[above,purple,opacity=1] {\scriptsize 2}
			($(v1)+(1.2*\scale,0.2*\scale)$);
			
			\node[purple] at ($(v1)+(1.75*\scale,0.1*\scale)$) {$\dots$};

			\draw[-,purple,opacity=1]  ($(Afg)+(-0.5*\scale,0.1*\scale)$) to [bend right=3] node[above,purple,opacity=1] {\scriptsize $p_1$} ($(v1)+(1.16*\scale,-0.23*\scale)$);
			
			\node[black] at ($(v1)+(1.25*\scale,-0.35*\scale)$) {\scriptsize $(i_1)$};
			
			\node[purple] at ($(v1)+(1.75*\scale,-0.35*\scale)$) {$\dots$};
			
			\draw[-,purple,opacity=1]  ($(Afg)+(-0.5*\scale,-0.1*\scale)$) to [bend left=16] node[above,purple,opacity=1] {\scriptsize $p_k$} ($(v1)+(0.84*\scale,-0.68*\scale)$);

			\node[black] at ($(v1)+(0.97*\scale,-0.8*\scale)$) {\scriptsize $(i_k)$};

			\node[purple] at ($(v1)+(1.75*\scale,-0.8*\scale)$) {$\dots$};
			
			\begin{scope}
			\draw[color=purple] ($(Afg)+(-0.5*\scale,-0.4*\scale)$)
			to[out=-135,in=0] ($(v1) + (0.5*\scale,-1.25*\scale)$)
			to[out=180,in=-90] ($(v1) + (-0.5*\scale,-0.4*\scale)$) 
			to[out=90,in=180] ($(v1) + (-0.2*\scale,0.2*\scale)$) 
			;
			\end{scope}

			\node[purple] at ($(v1)+(1.75*\scale,-1*\scale)$) {\scriptsize $t$};

			\draw[-,purple,opacity=1]  ($(Afg)+(-0.2*\scale,-0.4*\scale)$) to node[left,purple,opacity=1] {\scriptsize $t+1$} ($(Afg)+(-0.2*\scale,-1*\scale)$);
			
			\filldraw[color=purple] ($(Afg)+(-0.2*\scale,-1*\scale)$) circle [radius=0.02*\scale] node {};
			
			\node[purple] at ($(Afg)+(0*\scale,-0.75*\scale)$) {\scriptsize $\dots$};
			
			\draw[-,purple,opacity=1]  ($(Afg)+(0.5*\scale,-0.4*\scale)$) to node[left,purple,opacity=1] {\scriptsize $t+(t-k)\;$} ($(Afg)+(1.5*\scale,-1*\scale)$);
			
			\filldraw[color=purple] ($(Afg)+(1.5*\scale,-1*\scale)$) circle [radius=0.02*\scale] node {};
			
			\node[black] at ($(Afg)+(-0.2*\scale,-1.13*\scale)$) {\scriptsize $(j_1)$};
			
			\node[black] at ($(Afg)+(1.5*\scale,-1.13*\scale)$) {\scriptsize $(j_{t-k})$};

			\end{tikzpicture}
		\caption{The production $\nu_{1,\gamma,f,g}$.}\label{fig_I}
	\end{figure}
	
	\textbf{Type II. }
	Production of type II are used at the last step of a derivation such a on Fig. \ref{fig_new_der}; returning to the metaphor they allow one to finish sewing up and to finally connect a patch (that has grown by applying productions of types I and III) with the edges of a hole. 
	\\
	We recall that $\rho=A\to R$ such that $\delta(\rho)=e_0\in E_R$ and $lab(e_0)=A$. Let $f_0:[1,t]\to[1,t]$ be a function which is defined by the following relation: $f_0(i)=j\Leftrightarrow att_R(e_0)(i)=ext_R(j)$. Obviously, $f_0$ is a partial bijection. We set
	\begin{itemize}
		\item $V^\prime:=V_R$;
		\item $E^\prime :=E_G\setminus\{e_0\}$;
		\item $att^\prime=att_R|_{E^\prime}$, $lab^\prime=lab_R|_{E^\prime}$
		\item Let $k$ be equal to $|Dom(f_0)|$ and let $[1,t]\setminus Ran(f_0)=\{j_1,\dots,j_{t-k}\}_{IO}$. Then $ext^\prime=att_R(e_0)ext_R(j_1)\allowbreak\dots ext_R(j_{t-k})$.
	\end{itemize}
	We set $R^\prime=\langle V^\prime, E^\prime, att^\prime,lab^\prime,ext^\prime\rangle$. Finally, we introduce ${\nu_{2,\rho}:=(A,f_0,id)\to R^\prime}$ and set ${\delta(\nu_{2,\rho}):=e}$ for some chosen $e\in E_{R^\prime}$ (note that there is one since $|E_G|\ge 2$). Here $id:[1,t]\to [1,t]$ is the identity function. See Figure \ref{fig_II}.
	\begin{remark}
		$type(R^\prime)=2t-k=2\cdot type(A)-|Dom(f_0)|=2\cdot type(A)-|Dom(id\circ f_0)|=type((A,f_0,id))$. Thus the production is defined correctly.
	\end{remark}
	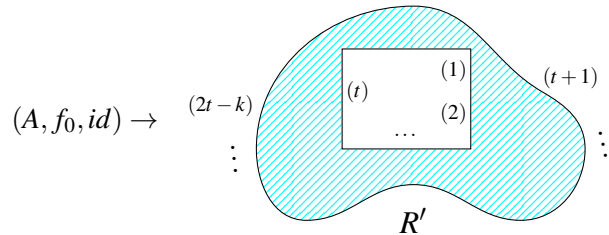
\begin{figure}
		\centering
			\begin{tikzpicture}
				\renewcommand{\scale}{1.9}
				\node (v0) at (0.5*\scale,0.5*\scale) {$(A,f_0,id)\to$};
				\node (v1) at (2.8*\scale,1*\scale) {};
				\begin{scope}[fill opacity=1]
				\filldraw[pattern=north east lines,pattern color=cyan] ($(v1)+(0*\scale,0.3*\scale)$) 
				to[out=0,in=150] ($(v1) + (0.9*\scale,-0.3*\scale)$) 
				to[out=-30,in=90] ($(v1) + (1.2*\scale,-0.7*\scale)$)
				to[out=-90,in=0] ($(v1) + (0.75*\scale,-1.2*\scale)$)
				to[out=180,in=0] ($(v1) + (0*\scale,-0.95*\scale)$)
				to[out=180,in=0] ($(v1) + (-0.75*\scale,-1.2*\scale)$)
				to[out=180,in=-90] ($(v1) + (-1.1*\scale,-0.7*\scale)$)
				to[out=90,in=180] ($(v1)+(0*\scale,0.3*\scale)$) 
				%%to[out=,in=] ($(v1) + (*\scale,*\scale)$)
				;
				\end{scope}
				\filldraw[color=white] ($(v1)+(-0.5*\scale,-0.7*\scale)$) rectangle ($(v1)+(0.4*\scale, 0*\scale)$);
				\draw[color=black] ($(v1)+(-0.5*\scale,-0.7*\scale)$) rectangle ($(v1)+(0.4*\scale,0*\scale)$);
				\node (R) at ($(v1)+(0*\scale,-1.2*\scale)$) {\large $R^\prime$};
				
				\node (e1) at ($(v1)+(1.1*\scale,-0.2*\scale)$) {\scriptsize $(t+1)$};
				\draw[-,black,opacity=0]  (e1) to [bend left=5] node[pos=0.55,sloped,above,black,opacity=1] {\dots} ($(v1)+(1.3*\scale,-1*\scale)$);
				\node (e2) at ($(v1)+(-1.35*\scale,-0.4*\scale)$) {\scriptsize $(2t-k)$};
				\draw[-,black,opacity=0]  (e2) to [bend right=5]
				node[pos=0.55,sloped,above,black,opacity=1] {\dots}
				($(v1)+(-1.3*\scale,-1*\scale)$);

				\node at ($(v1)+(0.28*\scale,-0.15*\scale)$) {\scriptsize $(1)$};
				
				\node at ($(v1)+(0.28*\scale,-0.45*\scale)$) {\scriptsize $(2)$};

				\node at ($(v1)+(-0.05*\scale,-0.6*\scale)$) {\scriptsize $\dots$};
				
				\node at ($(v1)+(-0.39*\scale,-0.3*\scale)$) {\scriptsize $(t)$};

			\end{tikzpicture}
		\caption{The production $\nu_{2,\rho}$.}\label{fig_II}
	\end{figure}
	\textbf{Type III. }
	Productions of type III serve to make a ``sewing up step'' from inside to outside. $f,g$ that specify relations between external nodes on the current step, on the next step and on the final step are changed by related functions $f^\prime,g^\prime$.
	\\
	Let $f,g,f^\prime,g^\prime:[1,t]\to[1,t]$ be partial bijections. 
	The whole procedure described below can be done if $g^\prime\circ f^\prime = g$.
	\newpage
	We set
	\begin{itemize}
		\item $V^\prime:=V_R\cup\{u_1,\dots,u_{t-k^\prime}\}$ such that $u_1,\dots, u_{t-k^\prime}$ are new nodes ($k^\prime$ is defined below);
		\item $E^\prime :=E_R\setminus\{e_0\}\cup\{e^\prime\}$ where $e_0=\delta(\rho)$ and $e^\prime$ is new;
		\item For $e\in E_R\setminus\{e_0\}$ we set $att^\prime(e):=att_R(e)$ and $lab^\prime(e)=lab_R(e)$;
		\item For $e^\prime$ we set $att^\prime(e^\prime):=ext_Ru_1\dots u_{t-k^\prime}$;
		\item $lab^\prime(e^\prime):=(A,f^\prime,g^\prime)$;
		\item Let 
		\begin{itemize}
			\item $M_1:=Dom(g)\setminus Ran(f)=\{s_1,\dots,s_p\}_{IO}$ (nodes of the current step that will be external at the last step but that were not taken into account at the previous step);
			\item $M_2:=[1,t]\setminus Ran (g\circ f)=\{j_1,\dots,j_{t-k}\}_{IO}$;
			\item $g(s_i)=j_{l_i}$;
			\item Let $M_2\setminus g(M_1)=\{j_{x_1},\dots,j_{x_{t-k-p}}\}$ where $\{x_1,\dots,x_{t-k-p}\}$ is index-ordered; here we note that $g(M_1)=Ran(g)\setminus Ran(g\circ f)$, thus $g(M_1)\subseteq M_2$.
		\end{itemize}
		Then we set $k^\prime = k + p$ and
		\begin{itemize}
			\item $ext^\prime(i):=att_R(e_0)(i)$ for $i=1,\dots,t$;
			\item $ext^\prime(t+l_i)=ext_R(s_i)$ for $i=1,\dots,p$;
			\item $ext^\prime(t+x_i):=u_{i}$ for $i=1,\dots,t-k^\prime$.
		\end{itemize}
	\end{itemize}
	We set $R^\prime=\langle V^\prime, E^\prime, att^\prime,lab^\prime,ext^\prime\rangle$. Then ${\nu_{3,\rho,f,g,f^\prime,g^\prime}:=(A,f,g)\to R^\prime}$. We put ${\delta(\nu_{3,\rho,f,g,f^\prime,g^\prime})=e}$ for some $e\in E_{R}\setminus\{e_0\}$. The production is illustrated on Figure \ref{fig_III}.
	\begin{remark}
		Again, we prove well-definedness of this production:
		\begin{itemize}
			\item Firstly, we check that $type(R^\prime)=type((A,f,g))$: $|ext^\prime|=t+p+t-k^\prime=2t+p-k-p=2\cdot type(A)-|Ran(g\circ f)|=type((A,f,g))$;
			\item Secondly, we have to check that $|att(e^\prime)|=type((a,f^\prime,g^\prime))$: $|att(e^\prime)|=t+t-k^\prime=t+|M_2\setminus g(M_1)|$. Note that $M_2\setminus g(M_1) = \left([1,t]\setminus Ran(g\circ f)\right)\setminus \left(Ran(g)\setminus Ran(g\circ f)\right)=[1,t]\setminus (Ran(g)\cup Ran(g\circ f))=[1,t]\setminus Ran(g)$. This yields $|att(e^\prime)|=2t-|Ran(g)|=2t-|Ran(g^\prime\circ f^\prime)|=type((A,f^\prime,g^\prime))$.
		\end{itemize}
	\end{remark}
	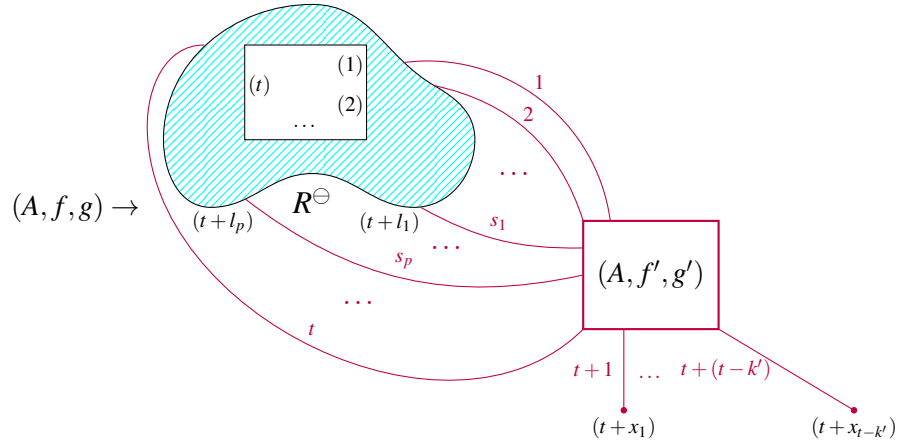
\begin{figure}
		\centering
		\begin{tikzpicture}
		\renewcommand{\scale}{1.8}
		\node (v0) at (0.5*\scale,0.5*\scale) {$(A,f,g)\to$};
		\node (v1) at (2.25*\scale,1.7*\scale) {};
		\begin{scope}[fill opacity=1]
		\filldraw[pattern=north east lines,pattern color=cyan] ($(v1)+(0*\scale,0.3*\scale)$) 
		to[out=0,in=150] ($(v1) + (0.9*\scale,-0.3*\scale)$) 
		to[out=-30,in=90] ($(v1) + (1.2*\scale,-0.7*\scale)$)
		to[out=-90,in=0] ($(v1) + (0.75*\scale,-1.2*\scale)$)
		to[out=180,in=0] ($(v1) + (0*\scale,-0.95*\scale)$)
		to[out=180,in=0] ($(v1) + (-0.75*\scale,-1.2*\scale)$)
		to[out=180,in=-90] ($(v1) + (-1.1*\scale,-0.7*\scale)$)
		to[out=90,in=180] ($(v1)+(0*\scale,0.3*\scale)$) 
		%%to[out=,in=] ($(v1) + (*\scale,*\scale)$)
		;
		\end{scope}
		\filldraw[color=white] ($(v1)+(-0.5*\scale,-0.7*\scale)$) rectangle ($(v1)+(0.4*\scale, 0*\scale)$);
		\draw[color=black] ($(v1)+(-0.5*\scale,-0.7*\scale)$) rectangle ($(v1)+(0.4*\scale,0*\scale)$);
		
		\node at ($(v1)+(0.28*\scale,-0.15*\scale)$) {\scriptsize $(1)$};
		
		\node at ($(v1)+(0.28*\scale,-0.45*\scale)$) {\scriptsize $(2)$};
		
		\node at ($(v1)+(-0.05*\scale,-0.6*\scale)$) {\scriptsize $\dots$};
		
		\node at ($(v1)+(-0.39*\scale,-0.3*\scale)$) {\scriptsize $(t)$};

		\node (Afg) at ($(v1)+(2.5*\scale,-1.7*\scale)$) { $(A,f^\prime,g^\prime)$};
		\draw[color=purple, text = black,thick] ($(Afg)+(-0.5*\scale,-0.4*\scale)$) rectangle ($(Afg)+(0.5*\scale,0.4*\scale)$);
		
		\draw[-,purple,opacity=1] ($(Afg)+(-0.3*\scale,0.4*\scale)$) to [bend right=45] 
		node[above,purple,opacity=1] {\scriptsize 1}
		($(v1)+(0.67*\scale,-0.13*\scale)$);

		\draw[-,purple,opacity=1] ($(Afg)+(-0.5*\scale,0.4*\scale)$) to [bend right=30] node[above,purple,opacity=1] {\scriptsize 2} ($(v1)+(0.9*\scale,-0.3*\scale)$);

		\node[purple] at ($(v1)+(1.5*\scale,-0.95*\scale)$) {$\dots$};
		
		\draw[-,purple,opacity=1] ($(Afg)+(-0.5*\scale,0.2*\scale)$) to [bend left=15] node[above,purple,opacity=1] {\scriptsize $s_1$} ($(v1)+(0.8*\scale,-1.2*\scale)$);
		
		\node at ($(v1)+(0.57*\scale,-1.31*\scale)$) {\scriptsize $(t+l_1)$};
		
		\node[purple] at ($(v1)+(1*\scale,-1.5*\scale)$) {$\dots$};
		
		\draw[-,purple,opacity=1] ($(Afg)+(-0.5*\scale,0*\scale)$) to [bend left=25] node[above,purple,opacity=1] {\scriptsize $s_p$} ($(v1)+(-0.5*\scale,-1.135*\scale)$);
		
		\node at ($(v1)+(-0.66*\scale,-1.31*\scale)$) {\scriptsize $(t+l_p)$};
		
		\begin{scope}
		\draw[color=purple] ($(Afg)+(-0.5*\scale,-0.4*\scale)$)
		to[out=-135,in=-90] ($(v1) + (-1.22*\scale,-0.6*\scale)$)
		to[out=90,in=180] ($(v1) + (-0.8*\scale,0*\scale)$)
		;
		\end{scope}
		
		\node[purple] at ($(v1)+(0.35*\scale,-1.9*\scale)$) {$\dots$};
		
		\node[purple] at ($(v1)+(0*\scale,-2.1*\scale)$) {\scriptsize $t$};

		\draw[-,purple,opacity=1]  ($(Afg)+(-0.2*\scale,-0.4*\scale)$) to node[left,purple,opacity=1] {\scriptsize $t+1$} ($(Afg)+(-0.2*\scale,-1*\scale)$);
		
		\filldraw[color=purple] ($(Afg)+(-0.2*\scale,-1*\scale)$) circle [radius=0.02*\scale] node {};
		
		\node[purple] at ($(Afg)+(0*\scale,-0.75*\scale)$) {\scriptsize $\dots$};
		
		\draw[-,purple,opacity=1]  ($(Afg)+(0.5*\scale,-0.4*\scale)$) to node[left,purple,opacity=1] {\scriptsize $t+(t-k^\prime)\;$} ($(Afg)+(1.5*\scale,-1*\scale)$);
		
		\filldraw[color=purple] ($(Afg)+(1.5*\scale,-1*\scale)$) circle [radius=0.02*\scale] node {};
		
		\node[black] at ($(Afg)+(-0.2*\scale,-1.13*\scale)$) {\scriptsize $(t+x_1)$};
		
		\node[black] at ($(Afg)+(1.5*\scale,-1.13*\scale)$) {\scriptsize $(t+x_{t-k^\prime})$};
		
		\node (R) at ($(v1)+(0*\scale,-1.15*\scale)$) {\large $R^\ominus$};
		\end{tikzpicture}
		\caption{The production $\nu_{3,\rho,f,g,f^\prime,g^\prime}$. Here $R^\ominus$ denotes a graph $R$ without $\delta(\rho)$.}\label{fig_III}
	\end{figure}
	
	We say that $P_2$ is obtained from $P_1$ by adding all possible productions of all the above types. Obviously, their number is finite (each production is defined by at most four partial functions from $[1,t]$ to $[1,t]$).
	\begin{lemma}\label{lem_rec}
		The grammar $HGr_1=\langle N^\prime, \Sigma, P_2,S\rangle$ is equivalent to $HGr=\langle N,\Sigma,P,S\rangle$.
	\end{lemma}
	\begin{proof}
		Informally, it suffices to notice that the design of new productions allows one to invert productions in the way which is shown on Figure \ref{fig_rederivation}. Example \ref{app_ex} provides an illustration of the formal proof, which is given below.
		
		Firstly, we check that $L(HGr)\subseteq L(HGr_1)$. In order to do that, we show how to model a branch of a derivation that consists of rules of the form $\rho_i$ and finishes with a rule $\gamma_j$ by means of $HGr_1$. Let $A\underset{\rho_{i_1}}{\Rightarrow} H_1\underset{\rho_{i_2}}{\Rightarrow} \dots \underset{\rho_{i_m}}{\Rightarrow} H_m \underset{\gamma_j}{\Rightarrow} H$ be a $\delta$-derivation in $HGr$ where $H_k=H_{k-1}[R_k/\delta(\rho_{i_{k-1}})], 1\le k < m$ and $H=H_m[G/\delta(\rho_{i_{m}})]$. It is convenient to put $H_0:=\circledcirc(A)$ and say that $\delta(\rho_{i_0}):=e_0$ where $E_{H_0}=\{e_0\}$. For the sake of simplyfing notations we write $\rho_1,\dots,\rho_m$ instead of $\rho_{i_1},\dots,\rho_{i_m}$ and $\gamma$ instead of $\gamma_j$.
		
		Let $G_k$ be defined inductively as follows: $G_{m+1}=G$; $G_k=R_k[G_{k+1}/\delta(\rho_k)]$. 
		Note that
		\[
		\begin{array}{lcl}
		H&=&H_m[G/\delta(\rho_m)]=\\&=&H_{m-1}[R_{m}/\delta(\rho_{m-1})][G/\delta(\rho_m)]=\\&\dots&\\&=&R_1[R_2/\delta(\rho_1)]\dots[R_m/\delta(\rho_{m-1})][G_{m+1}/\delta(\rho_m)]=\\&=&R_1[R_2/\delta(\rho_1)]\dots[G_{m}/\delta(\rho_{m-1})]=\\&\dots&\\&=&R_1[G_2/\delta(\rho_1)]=\\&=&G_1.\\
		\end{array}
		\]
		
		To prove the claimed inclusion, it suffices to show then that $A\underset{HGr_1}{\Rightarrow} G_1$.
		
		Observe that $H=H_{i-1}[R_i[G_{i+1}/\delta(\rho_i)]/\delta(\rho_{i-1})],\quad i=1,\dots,m$. Then $G_{i+1}$ can be considered as a subgraph of $G_i:=R_i[G_{i+1}/\delta(\rho_i)]$, which is a subgraph of $H$. Let $f_i$ be a partial function defined by the following correspondence: $f_i(j)=k\Leftrightarrow ext_{G_{i+1}}(j)=ext_{G_i}(k)$. Similarly we define $g_i$: $g_i(j)=k\Leftrightarrow ext_{G_i}(j)=ext_{H}(k)$. It follows from these definitions that $g_i\circ f_i = g_{i-1}$.
		
		Let $\sigma_i=A\to G_{i+1}$ and $G^\prime_i:=rhs(\nu_{1,\sigma_i,f_i,g_i})$.		
		We prove by the reverse induction on $p=m,\dots,1$ that then $A\underset{HGr_1}{\Rightarrow}G^\prime_p$. 
		
		\textbf{Basis.} Since $\sigma_m=A\to G=\gamma$ is not a recursive $A$-production, we added $\nu_{1,\sigma_m,f_m,g_m}$ in our grammar. This completes the basis case.
		
		\textbf{Step.} We assume by the induction hypothesis that $A\underset{HGr_1}{\Rightarrow}G^\prime_p$ and our aim is to prove the same for $p-1$. $G^\prime_p$ contains an edge labeled by $(A,f_p,g_p)$. Then a type III rule is applied to this edge: $\nu_{3,\rho_p,f_p,g_p,f_{p-1},g_{p-1}}$. The design of such rules allows us to perform exactly the substitution of $G_{p+1}$ into $R_p$. 
		
		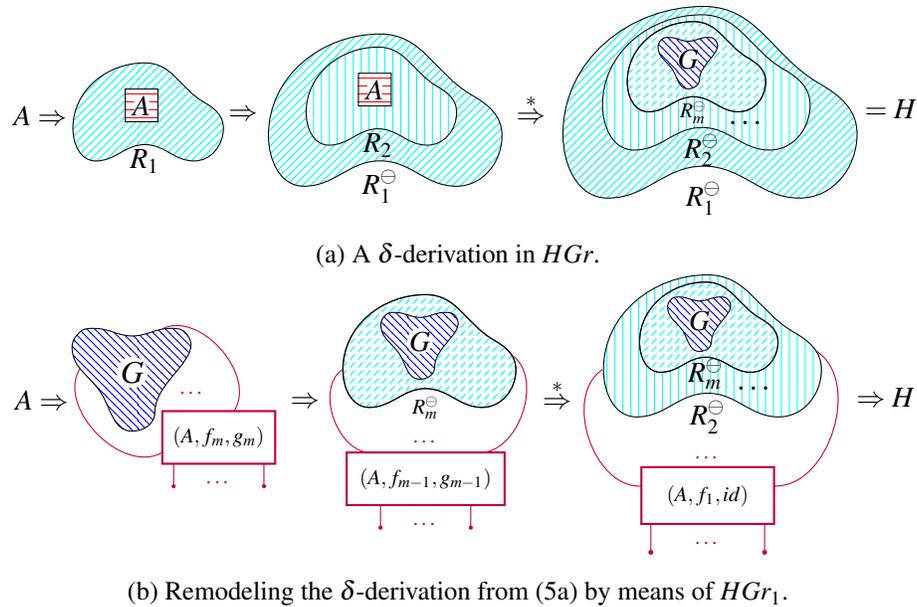
\begin{figure}
			\centering
			\subfloat[A $\delta$-derivation in $HGr$.]{
			\label{fig_source_der}
			\begin{tikzpicture}
			\renewcommand{\scale}{0.87}
			\node (v0) at (0.5*\scale,0.5*\scale) {$A\Rightarrow$};
			\node (v1) at ($(v0)+(1.6*\scale,0.5*\scale)$) {};
			\begin{scope}[fill opacity=1]
			\filldraw[pattern=north east lines,pattern color=cyan] ($(v1)+(0*\scale,0.3*\scale)$) 
			to[out=0,in=150] ($(v1) + (0.9*\scale,-0.3*\scale)$) 
			to[out=-30,in=90] ($(v1) + (1.2*\scale,-0.7*\scale)$)
			to[out=-90,in=0] ($(v1) + (0.75*\scale,-1.2*\scale)$)
			to[out=180,in=0] ($(v1) + (0*\scale,-0.95*\scale)$)
			to[out=180,in=0] ($(v1) + (-0.75*\scale,-1.2*\scale)$)
			to[out=180,in=-90] ($(v1) + (-1.1*\scale,-0.7*\scale)$)
			to[out=90,in=180] ($(v1)+(0*\scale,0.3*\scale)$) 
			%%to[out=,in=] ($(v1) + (*\scale,*\scale)$)
			;
			\end{scope}
			\filldraw[color=white] ($(v1)+(-0.3*\scale,-0.6*\scale)$) rectangle ($(v1)+(0.2*\scale,-0.1*\scale)$);
			\filldraw[pattern=horizontal lines,pattern color=red] ($(v1)+(-0.3*\scale,-0.6*\scale)$) rectangle ($(v1)+(0.2*\scale,-0.1*\scale)$);
			\node (R) at ($(v1)+(0*\scale,-1.2*\scale)$) {$R_1$};
			\filldraw[color=white,text=black] ($(v1)+(-0.05*\scale,-0.35*\scale)$) circle [radius=0.12*\scale] node {$A$};

			\node (v01) at (3.6*\scale,0.5*\scale) {$\Rightarrow$};
			\renewcommand{\scale}{1.35}
			\node (v11) at ($(v01)+(1.35*\scale,0.5*\scale)$) {};
			\begin{scope}[fill opacity=1]
			\filldraw[pattern=north east lines,pattern color=cyan] ($(v11)+(0*\scale,0.3*\scale)$) 
			to[out=0,in=150] ($(v11) + (0.9*\scale,-0.3*\scale)$) 
			to[out=-30,in=90] ($(v11) + (1.2*\scale,-0.7*\scale)$)
			to[out=-90,in=0] ($(v11) + (0.75*\scale,-1.2*\scale)$)
			to[out=180,in=0] ($(v11) + (0*\scale,-0.95*\scale)$)
			to[out=180,in=0] ($(v11) + (-0.75*\scale,-1.2*\scale)$)
			to[out=180,in=-90] ($(v11) + (-1.1*\scale,-0.7*\scale)$)
			to[out=90,in=180] ($(v11)+(0*\scale,0.3*\scale)$) 
			%%to[out=,in=] ($(v1) + (*\scale,*\scale)$)
			;
			\end{scope}
			\node (R1) at ($(v11)+(0*\scale,-1.2*\scale)$) {$R_1^\ominus$};
			
			\node (v12) at ($(3.65*\scale,0.8*\scale)$) {};
			\renewcommand{\scale}{0.87}
			\begin{scope}[fill opacity=1]
			\filldraw[color=white] ($(v12)+(0*\scale,0.3*\scale)$) 
			to[out=0,in=150] ($(v12) + (0.9*\scale,-0.3*\scale)$) 
			to[out=-30,in=90] ($(v12) + (1.2*\scale,-0.7*\scale)$)
			to[out=-90,in=0] ($(v12) + (0.75*\scale,-1.2*\scale)$)
			to[out=180,in=0] ($(v12) + (0*\scale,-0.95*\scale)$)
			to[out=180,in=0] ($(v12) + (-0.75*\scale,-1.2*\scale)$)
			to[out=180,in=-90] ($(v12) + (-1.1*\scale,-0.7*\scale)$)
			to[out=90,in=180] ($(v12)+(0*\scale,0.3*\scale)$) 
			%%to[out=,in=] ($(v1) + (*\scale,*\scale)$)
			;
			\end{scope}
			
			\begin{scope}[fill opacity=1]
			\filldraw[pattern=vertical lines,pattern color=cyan] ($(v12)+(0*\scale,0.3*\scale)$) 
			to[out=0,in=150] ($(v12) + (0.9*\scale,-0.3*\scale)$) 
			to[out=-30,in=90] ($(v12) + (1.2*\scale,-0.7*\scale)$)
			to[out=-90,in=0] ($(v12) + (0.75*\scale,-1.2*\scale)$)
			to[out=180,in=0] ($(v12) + (0*\scale,-0.95*\scale)$)
			to[out=180,in=0] ($(v12) + (-0.75*\scale,-1.2*\scale)$)
			to[out=180,in=-90] ($(v12) + (-1.1*\scale,-0.7*\scale)$)
			to[out=90,in=180] ($(v12)+(0*\scale,0.3*\scale)$) 
			%%to[out=,in=] ($(v1) + (*\scale,*\scale)$)
			;
			\end{scope}
			
			\filldraw[color=white] ($(v12)+(-0.3*\scale,-0.6*\scale)$) rectangle ($(v12)+(0.2*\scale,-0.1*\scale)$);
			\filldraw[pattern=horizontal lines,pattern color=red] ($(v12)+(-0.3*\scale,-0.6*\scale)$) rectangle ($(v12)+(0.2*\scale,-0.1*\scale)$);
			\node (R2) at ($(v12)+(0*\scale,-1.2*\scale)$) {$R_2$};
			\filldraw[color=white,text=black] ($(v12)+(-0.05*\scale,-0.35*\scale)$) circle [radius=0.12*\scale] node {$A$};

			\node (v02) at (8*\scale,0.6*\scale) {$\overset{\ast}{\Rightarrow}$};
			\renewcommand{\scale}{1.7}
			
			\node (v13) at ($(v02)+(1.35*\scale,0.5*\scale)$) {};
			\begin{scope}[fill opacity=1]
			\filldraw[pattern=north east lines,pattern color=cyan] ($(v13)+(0*\scale,0.3*\scale)$) 
			to[out=0,in=150] ($(v13) + (0.9*\scale,-0.3*\scale)$) 
			to[out=-30,in=90] ($(v13) + (1.2*\scale,-0.7*\scale)$)
			to[out=-90,in=0] ($(v13) + (0.75*\scale,-1.2*\scale)$)
			to[out=180,in=0] ($(v13) + (0*\scale,-0.95*\scale)$)
			to[out=180,in=0] ($(v13) + (-0.75*\scale,-1.2*\scale)$)
			to[out=180,in=-90] ($(v13) + (-1.1*\scale,-0.7*\scale)$)
			to[out=90,in=180] ($(v13)+(0*\scale,0.3*\scale)$) 
			%%to[out=,in=] ($(v1) + (*\scale,*\scale)$)
			;
			\end{scope}
			\node (R1) at ($(v13)+(0*\scale,-1.2*\scale)$) {$R_1^\ominus$};
			
			\node (v14) at ($(5.43*\scale,0.83*\scale)$) {};
			\renewcommand{\scale}{1.2}
			\begin{scope}[fill opacity=1]
			\filldraw[color=white] ($(v14)+(0*\scale,0.3*\scale)$) 
			to[out=0,in=150] ($(v14) + (0.9*\scale,-0.3*\scale)$) 
			to[out=-30,in=90] ($(v14) + (1.2*\scale,-0.7*\scale)$)
			to[out=-90,in=0] ($(v14) + (0.75*\scale,-1.2*\scale)$)
			to[out=180,in=0] ($(v14) + (0*\scale,-0.95*\scale)$)
			to[out=180,in=0] ($(v14) + (-0.75*\scale,-1.2*\scale)$)
			to[out=180,in=-90] ($(v14) + (-1.1*\scale,-0.7*\scale)$)
			to[out=90,in=180] ($(v14)+(0*\scale,0.3*\scale)$) 
			%%to[out=,in=] ($(v1) + (*\scale,*\scale)$)
			;
			\end{scope}
			
			\begin{scope}[fill opacity=1]
			\filldraw[pattern=vertical lines,pattern color=cyan] ($(v14)+(0*\scale,0.3*\scale)$) 
			to[out=0,in=150] ($(v14) + (0.9*\scale,-0.3*\scale)$) 
			to[out=-30,in=90] ($(v14) + (1.2*\scale,-0.7*\scale)$)
			to[out=-90,in=0] ($(v14) + (0.75*\scale,-1.2*\scale)$)
			to[out=180,in=0] ($(v14) + (0*\scale,-0.95*\scale)$)
			to[out=180,in=0] ($(v14) + (-0.75*\scale,-1.2*\scale)$)
			to[out=180,in=-90] ($(v14) + (-1.1*\scale,-0.7*\scale)$)
			to[out=90,in=180] ($(v14)+(0*\scale,0.3*\scale)$) 
			%%to[out=,in=] ($(v1) + (*\scale,*\scale)$)
			;
			\end{scope}
			
			\node at ($(v14)+(0*\scale,-1.2*\scale)$) {$R_2^\ominus$};
			
			\node at ($(v14)+(0.5*\scale,-0.9*\scale)$) {\large\bf $\dots$};
			
			\node (v15) at ($(7.6*\scale,1.2*\scale)$) {};
			\renewcommand{\scale}{0.8}
			\begin{scope}[fill opacity=1]
			\filldraw[color=black, fill=white] ($(v15)+(0*\scale,0.3*\scale)$) 
			to[out=0,in=150] ($(v15) + (0.9*\scale,-0.3*\scale)$) 
			to[out=-30,in=90] ($(v15) + (1.2*\scale,-0.7*\scale)$)
			to[out=-90,in=0] ($(v15) + (0.75*\scale,-1.2*\scale)$)
			to[out=180,in=0] ($(v15) + (0*\scale,-0.95*\scale)$)
			to[out=180,in=0] ($(v15) + (-0.75*\scale,-1.2*\scale)$)
			to[out=180,in=-90] ($(v15) + (-1.1*\scale,-0.7*\scale)$)
			to[out=90,in=180] ($(v15)+(0*\scale,0.3*\scale)$) 
			%%to[out=,in=] ($(v1) + (*\scale,*\scale)$)
			;
			\end{scope}

			\begin{scope}[fill opacity=1]
			\filldraw[pattern=NELines,pattern color=cyan] ($(v15)+(0*\scale,0.3*\scale)$) 
			to[out=0,in=150] ($(v15) + (0.9*\scale,-0.3*\scale)$) 
			to[out=-30,in=90] ($(v15) + (1.2*\scale,-0.7*\scale)$)
			to[out=-90,in=0] ($(v15) + (0.75*\scale,-1.2*\scale)$)
			to[out=180,in=0] ($(v15) + (0*\scale,-0.95*\scale)$)
			to[out=180,in=0] ($(v15) + (-0.75*\scale,-1.2*\scale)$)
			to[out=180,in=-90] ($(v15) + (-1.1*\scale,-0.7*\scale)$)
			to[out=90,in=180] ($(v15)+(0*\scale,0.3*\scale)$) 
			%%to[out=,in=] ($(v1) + (*\scale,*\scale)$)
			;
			\end{scope}
			
			\node at ($(v15)+(0*\scale,-1.2*\scale)$) {\scriptsize $R_m^\ominus$};

			\node (v16) at (11*\scale,1.68*\scale) {};
			\renewcommand{\scale}{0.5}
			\begin{scope}[fill opacity=1]
			\filldraw[color=white] ($(v16)+(-0.2*\scale,0.2*\scale)$) 
			to[out=45,in=180] ($(v16) + (0*\scale,0.25*\scale)$) 
			to[out=0,in=180] ($(v16) + (0.5*\scale,0.15*\scale)$)
			to[out=0,in=180] ($(v16) + (1*\scale,0.25*\scale)$)
			to[out=0,in=135] ($(v16)+(1.2*\scale,0.2*\scale)$)
			to[out=-45,in=90] ($(v16)+(0.8*\scale,-0.866*\scale)$)
			to[out=270,in=0] ($(v16)+(0.5*\scale,-1.116*\scale)$)
			to[out=180,in=-60] ($(v16)+(0.1*\scale,-0.634*\scale)$)
			to[out=120,in=270] ($(v16)+(-0.3*\scale,0*\scale)$)
			to[out=90,in=-135] ($(v16)+(-0.2*\scale,0.2*\scale)$)
			;
			
			\filldraw[pattern=north west lines,pattern color=blue] ($(v16)+(-0.2*\scale,0.2*\scale)$) 
			to[out=45,in=180] ($(v16) + (0*\scale,0.25*\scale)$) 
			to[out=0,in=180] ($(v16) + (0.5*\scale,0.15*\scale)$)
			to[out=0,in=180] ($(v16) + (1*\scale,0.25*\scale)$)
			to[out=0,in=135] ($(v16)+(1.2*\scale,0.2*\scale)$)
			to[out=-45,in=90] ($(v16)+(0.8*\scale,-0.866*\scale)$)
			to[out=270,in=0] ($(v16)+(0.5*\scale,-1.116*\scale)$)
			to[out=180,in=-60] ($(v16)+(0.1*\scale,-0.634*\scale)$)
			to[out=120,in=270] ($(v16)+(-0.3*\scale,0*\scale)$)
			to[out=90,in=-135] ($(v16)+(-0.2*\scale,0.2*\scale)$)
			;
			\end{scope}		
			\filldraw[color=white,text=black] ($(v16)+(0.5*\scale,-0.334*\scale)$) circle [radius=0.25*\scale] node {$G$};
			
			\renewcommand{\scale}{0.87}
			\node at ($(13.5*\scale,0.6*\scale)$) {$=H$};
			
			\end{tikzpicture}
		}

		\subfloat[Remodeling the $\delta$-derivation from (\ref{fig_source_der}) by means of $HGr_1$.]{
			\label{fig_new_der}
			\begin{tikzpicture}
			\renewcommand{\scale}{1}
			\node (v0) at (0.5*\scale,0.5*\scale) {$A\Rightarrow$};
			\node (v1) at (1.2*\scale,1.2*\scale) {};
			\begin{scope}[fill opacity=1]
			\filldraw[pattern=north west lines,pattern color=blue] ($(v1)+(-0.2*\scale,0.2*\scale)$) 
			to[out=45,in=180] ($(v1) + (0*\scale,0.25*\scale)$) 
			to[out=0,in=180] ($(v1) + (0.5*\scale,0.15*\scale)$)
			to[out=0,in=180] ($(v1) + (1*\scale,0.25*\scale)$)
			to[out=0,in=135] ($(v1)+(1.2*\scale,0.2*\scale)$)
			to[out=-45,in=90] ($(v1)+(0.8*\scale,-0.866*\scale)$)
			to[out=270,in=0] ($(v1)+(0.5*\scale,-1.116*\scale)$)
			to[out=180,in=-60] ($(v1)+(0.1*\scale,-0.634*\scale)$)
			to[out=120,in=270] ($(v1)+(-0.3*\scale,0*\scale)$)
			to[out=90,in=-135] ($(v1)+(-0.2*\scale,0.2*\scale)$)
			;
			\end{scope}		
			\filldraw[color=white,text=black] ($(v1)+(0.5*\scale,-0.334*\scale)$) circle [radius=0.2*\scale] node {\large $G$};
			
			\node (Afg) at ($(v1)+(1.65*\scale,-1.2*\scale)$) {\scriptsize $(A,f_m,g_m)$};
			\draw[color=purple, text = black,thick] ($(Afg)+(-0.75*\scale,-0.35*\scale)$) rectangle ($(Afg)+(0.75*\scale,0.35*\scale)$);
			
			\draw[-,purple,opacity=1]  ($(Afg)+(0.2*\scale,0.35*\scale)$) to [bend right=75] 
			($(v1)+(0.9*\scale,0.245*\scale)$);
			
			\node [purple] at ($(v1)+(1.3*\scale,-0.6*\scale)$) {\scriptsize \dots};
			
			\draw[-,purple,opacity=1]  ($(Afg)+(-0.75*\scale,-0.2*\scale)$) to [bend left=75] 
			($(v1)+(-0.2*\scale,-0.29*\scale)$);
			
			\draw[-,purple,opacity=1] ($(Afg)+(-0.6*\scale,-0.35*\scale)$) to ($(Afg)+(-0.6*\scale,-0.65*\scale)$);
			\filldraw[color=purple] ($(Afg)+(-0.6*\scale,-0.65*\scale)$) circle [radius=0.02*\scale] node {};
			
			\node [purple] at ($(Afg)+(0*\scale,-0.6*\scale)$) {\scriptsize \dots};
			
			\draw[-,purple,opacity=1] ($(Afg)+(0.6*\scale,-0.35*\scale)$) to ($(Afg)+(0.6*\scale,-0.65*\scale)$);
			\filldraw[color=purple] ($(Afg)+(0.6*\scale,-0.65*\scale)$) circle [radius=0.02*\scale] node {};

			\node (v02) at (4*\scale,0.465*\scale) {$\Rightarrow$};
			
			\node (v15) at ($(5.6*\scale,1.6*\scale)$) {};
			\renewcommand{\scale}{1}
			\begin{scope}[fill opacity=1]
			\filldraw[color=black, fill=white] ($(v15)+(0*\scale,0.3*\scale)$) 
			to[out=0,in=150] ($(v15) + (0.9*\scale,-0.3*\scale)$) 
			to[out=-30,in=90] ($(v15) + (1.2*\scale,-0.7*\scale)$)
			to[out=-90,in=0] ($(v15) + (0.75*\scale,-1.2*\scale)$)
			to[out=180,in=0] ($(v15) + (0*\scale,-0.95*\scale)$)
			to[out=180,in=0] ($(v15) + (-0.75*\scale,-1.2*\scale)$)
			to[out=180,in=-90] ($(v15) + (-1.1*\scale,-0.7*\scale)$)
			to[out=90,in=180] ($(v15)+(0*\scale,0.3*\scale)$) 
			%%to[out=,in=] ($(v1) + (*\scale,*\scale)$)
			;
			\end{scope}

			\begin{scope}[fill opacity=1]
			\filldraw[pattern=NELines,pattern color=cyan] ($(v15)+(0*\scale,0.3*\scale)$) 
			to[out=0,in=150] ($(v15) + (0.9*\scale,-0.3*\scale)$) 
			to[out=-30,in=90] ($(v15) + (1.2*\scale,-0.7*\scale)$)
			to[out=-90,in=0] ($(v15) + (0.75*\scale,-1.2*\scale)$)
			to[out=180,in=0] ($(v15) + (0*\scale,-0.95*\scale)$)
			to[out=180,in=0] ($(v15) + (-0.75*\scale,-1.2*\scale)$)
			to[out=180,in=-90] ($(v15) + (-1.1*\scale,-0.7*\scale)$)
			to[out=90,in=180] ($(v15)+(0*\scale,0.3*\scale)$) 
			%%to[out=,in=] ($(v1) + (*\scale,*\scale)$)
			;
			\end{scope}
			
			\node at ($(v15)+(0*\scale,-1.2*\scale)$) {\scriptsize $R_m^\ominus$};

			\node (v16) at (5.2*\scale,1.5*\scale) {};
			\renewcommand{\scale}{0.65}
			\begin{scope}[fill opacity=1]
			\filldraw[color=white] ($(v16)+(-0.2*\scale,0.2*\scale)$) 
			to[out=45,in=180] ($(v16) + (0*\scale,0.25*\scale)$) 
			to[out=0,in=180] ($(v16) + (0.5*\scale,0.15*\scale)$)
			to[out=0,in=180] ($(v16) + (1*\scale,0.25*\scale)$)
			to[out=0,in=135] ($(v16)+(1.2*\scale,0.2*\scale)$)
			to[out=-45,in=90] ($(v16)+(0.8*\scale,-0.866*\scale)$)
			to[out=270,in=0] ($(v16)+(0.5*\scale,-1.116*\scale)$)
			to[out=180,in=-60] ($(v16)+(0.1*\scale,-0.634*\scale)$)
			to[out=120,in=270] ($(v16)+(-0.3*\scale,0*\scale)$)
			to[out=90,in=-135] ($(v16)+(-0.2*\scale,0.2*\scale)$)
			;
			
			\filldraw[pattern=north west lines,pattern color=blue] ($(v16)+(-0.2*\scale,0.2*\scale)$) 
			to[out=45,in=180] ($(v16) + (0*\scale,0.25*\scale)$) 
			to[out=0,in=180] ($(v16) + (0.5*\scale,0.15*\scale)$)
			to[out=0,in=180] ($(v16) + (1*\scale,0.25*\scale)$)
			to[out=0,in=135] ($(v16)+(1.2*\scale,0.2*\scale)$)
			to[out=-45,in=90] ($(v16)+(0.8*\scale,-0.866*\scale)$)
			to[out=270,in=0] ($(v16)+(0.5*\scale,-1.116*\scale)$)
			to[out=180,in=-60] ($(v16)+(0.1*\scale,-0.634*\scale)$)
			to[out=120,in=270] ($(v16)+(-0.3*\scale,0*\scale)$)
			to[out=90,in=-135] ($(v16)+(-0.2*\scale,0.2*\scale)$)
			;
			\end{scope}		
			\filldraw[color=white,text=black] ($(v16)+(0.5*\scale,-0.334*\scale)$) circle [radius=0.25*\scale] node {$G$};
			
			\renewcommand{\scale}{1}
			\node (Afg1) at ($(v15)+(0*\scale,-2.15*\scale)$) {\scriptsize $(A,f_{m-1},g_{m-1})$};
			\draw[color=purple, text = black,thick] ($(Afg1)+(-1.05*\scale,-0.35*\scale)$) rectangle ($(Afg1)+(1.05*\scale,0.35*\scale)$);
			
			\draw[-,purple,opacity=1]  ($(Afg1)+(0.7*\scale,0.35*\scale)$) to [bend right=75] 
			($(v15)+(1.14*\scale,-0.5*\scale)$);
			
			\node [purple] at ($(Afg1)+(0*\scale,0.5*\scale)$) {\scriptsize \dots};
			
			\draw[-,purple,opacity=1]  ($(Afg1)+(-0.7*\scale,0.35*\scale)$) to [bend left=75] 
			($(v15)+(-1.11*\scale,-0.7*\scale)$);
			
			\draw[-,purple,opacity=1] ($(Afg1)+(-0.6*\scale,-0.35*\scale)$) to ($(Afg1)+(-0.6*\scale,-0.65*\scale)$);
			\filldraw[color=purple] ($(Afg1)+(-0.6*\scale,-0.65*\scale)$) circle [radius=0.02*\scale] node {};
			
			\node [purple] at ($(Afg1)+(0*\scale,-0.6*\scale)$) {\scriptsize \dots};
			
			\draw[-,purple,opacity=1] ($(Afg1)+(0.6*\scale,-0.35*\scale)$) to ($(Afg1)+(0.6*\scale,-0.65*\scale)$);
			\filldraw[color=purple] ($(Afg1)+(0.6*\scale,-0.65*\scale)$) circle [radius=0.02*\scale] node {};

			\renewcommand{\scale}{1}
			\node at (7.35*\scale,0.58*\scale) {$\overset{\ast}{\Rightarrow}$};
			\renewcommand{\scale}{1.7}
			
			\node (v14) at ($(5.5*\scale,1.05*\scale)$) {};
			\renewcommand{\scale}{1.26}
			\begin{scope}[fill opacity=1]
			\filldraw[color=white] ($(v14)+(0*\scale,0.3*\scale)$) 
			to[out=0,in=150] ($(v14) + (0.9*\scale,-0.3*\scale)$) 
			to[out=-30,in=90] ($(v14) + (1.2*\scale,-0.7*\scale)$)
			to[out=-90,in=0] ($(v14) + (0.75*\scale,-1.2*\scale)$)
			to[out=180,in=0] ($(v14) + (0*\scale,-0.95*\scale)$)
			to[out=180,in=0] ($(v14) + (-0.75*\scale,-1.2*\scale)$)
			to[out=180,in=-90] ($(v14) + (-1.1*\scale,-0.7*\scale)$)
			to[out=90,in=180] ($(v14)+(0*\scale,0.3*\scale)$) 
			%%to[out=,in=] ($(v1) + (*\scale,*\scale)$)
			;
			\end{scope}
			
			\begin{scope}[fill opacity=1]
			\filldraw[pattern=vertical lines,pattern color=cyan] ($(v14)+(0*\scale,0.3*\scale)$) 
			to[out=0,in=150] ($(v14) + (0.9*\scale,-0.3*\scale)$) 
			to[out=-30,in=90] ($(v14) + (1.2*\scale,-0.7*\scale)$)
			to[out=-90,in=0] ($(v14) + (0.75*\scale,-1.2*\scale)$)
			to[out=180,in=0] ($(v14) + (0*\scale,-0.95*\scale)$)
			to[out=180,in=0] ($(v14) + (-0.75*\scale,-1.2*\scale)$)
			to[out=180,in=-90] ($(v14) + (-1.1*\scale,-0.7*\scale)$)
			to[out=90,in=180] ($(v14)+(0*\scale,0.3*\scale)$) 
			%%to[out=,in=] ($(v1) + (*\scale,*\scale)$)
			;
			\end{scope}
			
			\node at ($(v14)+(0*\scale,-1.2*\scale)$) {$R_2^\ominus$};
			
			\node at ($(v14)+(0.5*\scale,-0.9*\scale)$) {\large\bf $\dots$};
			
			\node (v17) at ($(7.4*\scale,1.45*\scale)$) {};
			\renewcommand{\scale}{0.8}
			\begin{scope}[fill opacity=1]
			\filldraw[color=black, fill=white] ($(v17)+(0*\scale,0.3*\scale)$) 
			to[out=0,in=150] ($(v17) + (0.9*\scale,-0.3*\scale)$) 
			to[out=-30,in=90] ($(v17) + (1.2*\scale,-0.7*\scale)$)
			to[out=-90,in=0] ($(v17) + (0.75*\scale,-1.2*\scale)$)
			to[out=180,in=0] ($(v17) + (0*\scale,-0.95*\scale)$)
			to[out=180,in=0] ($(v17) + (-0.75*\scale,-1.2*\scale)$)
			to[out=180,in=-90] ($(v17) + (-1.1*\scale,-0.7*\scale)$)
			to[out=90,in=180] ($(v17)+(0*\scale,0.3*\scale)$) 
			%%to[out=,in=] ($(v1) + (*\scale,*\scale)$)
			;
			\end{scope}

			\begin{scope}[fill opacity=1]
			\filldraw[pattern=NELines,pattern color=cyan] ($(v17)+(0*\scale,0.3*\scale)$) 
			to[out=0,in=150] ($(v17) + (0.9*\scale,-0.3*\scale)$) 
			to[out=-30,in=90] ($(v17) + (1.2*\scale,-0.7*\scale)$)
			to[out=-90,in=0] ($(v17) + (0.75*\scale,-1.2*\scale)$)
			to[out=180,in=0] ($(v17) + (0*\scale,-0.95*\scale)$)
			to[out=180,in=0] ($(v17) + (-0.75*\scale,-1.2*\scale)$)
			to[out=180,in=-90] ($(v17) + (-1.1*\scale,-0.7*\scale)$)
			to[out=90,in=180] ($(v17)+(0*\scale,0.3*\scale)$) 
			%%to[out=,in=] ($(v1) + (*\scale,*\scale)$)
			;
			\end{scope}
			
			\node at ($(v17)+(0*\scale,-1.23*\scale)$) { $R_m^\ominus$};

			\node (v18) at (11.22*\scale,2.17*\scale) {};
			\renewcommand{\scale}{0.53}
			\begin{scope}[fill opacity=1]
			\filldraw[color=white] ($(v18)+(-0.2*\scale,0.2*\scale)$) 
			to[out=45,in=180] ($(v18) + (0*\scale,0.25*\scale)$) 
			to[out=0,in=180] ($(v18) + (0.5*\scale,0.15*\scale)$)
			to[out=0,in=180] ($(v18) + (1*\scale,0.25*\scale)$)
			to[out=0,in=135] ($(v18)+(1.2*\scale,0.2*\scale)$)
			to[out=-45,in=90] ($(v18)+(0.8*\scale,-0.866*\scale)$)
			to[out=270,in=0] ($(v18)+(0.5*\scale,-1.116*\scale)$)
			to[out=180,in=-60] ($(v18)+(0.1*\scale,-0.634*\scale)$)
			to[out=120,in=270] ($(v18)+(-0.3*\scale,0*\scale)$)
			to[out=90,in=-135] ($(v18)+(-0.2*\scale,0.2*\scale)$)
			;
			
			\filldraw[pattern=north west lines,pattern color=blue] ($(v18)+(-0.2*\scale,0.2*\scale)$) 
			to[out=45,in=180] ($(v18) + (0*\scale,0.25*\scale)$) 
			to[out=0,in=180] ($(v18) + (0.5*\scale,0.15*\scale)$)
			to[out=0,in=180] ($(v18) + (1*\scale,0.25*\scale)$)
			to[out=0,in=135] ($(v18)+(1.2*\scale,0.2*\scale)$)
			to[out=-45,in=90] ($(v18)+(0.8*\scale,-0.866*\scale)$)
			to[out=270,in=0] ($(v18)+(0.5*\scale,-1.116*\scale)$)
			to[out=180,in=-60] ($(v18)+(0.1*\scale,-0.634*\scale)$)
			to[out=120,in=270] ($(v18)+(-0.3*\scale,0*\scale)$)
			to[out=90,in=-135] ($(v18)+(-0.2*\scale,0.2*\scale)$)
			;
			\end{scope}		
			\filldraw[color=white,text=black] ($(v18)+(0.5*\scale,-0.334*\scale)$) circle [radius=0.25*\scale] node {$G$};
			
			\renewcommand{\scale}{1.26}
			\node (Afg2) at ($(v14)+(0*\scale,-2.03*\scale)$) {\scriptsize $(A,f_{1},id)$};
			\draw[color=purple, text = black,thick] ($(Afg2)+(-0.7*\scale,-0.3*\scale)$) rectangle ($(Afg2)+(0.7*\scale,0.3*\scale)$);
			
			\draw[-,purple,opacity=1]  ($(Afg2)+(0.7*\scale,0.1*\scale)$) to [bend right=75] 
			($(v14)+(1.14*\scale,-0.5*\scale)$);
			
			\node [purple] at ($(Afg2)+(0*\scale,0.4*\scale)$) {\scriptsize \dots};
			
			\draw[-,purple,opacity=1]  ($(Afg2)+(-0.7*\scale,0.1*\scale)$) to [bend left=75] 
			($(v14)+(-1.11*\scale,-0.7*\scale)$);
			
			\draw[-,purple,opacity=1] ($(Afg2)+(-0.6*\scale,-0.3*\scale)$) to ($(Afg2)+(-0.6*\scale,-0.6*\scale)$);
			\filldraw[color=purple] ($(Afg2)+(-0.6*\scale,-0.6*\scale)$) circle [radius=0.02*\scale] node {};
			
			\node [purple] at ($(Afg2)+(0*\scale,-0.5*\scale)$) {\scriptsize \dots};
			
			\draw[-,purple,opacity=1] ($(Afg2)+(0.6*\scale,-0.3*\scale)$) to ($(Afg2)+(0.6*\scale,-0.6*\scale)$);
			\filldraw[color=purple] ($(Afg2)+(0.6*\scale,-0.6*\scale)$) circle [radius=0.02*\scale] node {};

			\node at ($(9.3*\scale,0.42*\scale)$) {$\Rightarrow H$};
			\end{tikzpicture}
			}
			\caption{Illustration of the proof of Lemma \ref{lem_rec}. External nodes are not depicted here.}\label{fig_rederivation}
		\end{figure}
		
		%%\begin{figure}
		%%	\centering
		%%	\includegraphics[width=0.7\linewidth]{Figures/composition}
		%%	\caption{The composition of productions $\nu_{1,\gamma,f,g}$ and $\nu_{3,\rho,f,g,f^\prime,g^\prime}$.}\label{fig_composition}
		%%\end{figure}

		Therefore, $A\underset{HGr_1}{\Rightarrow}G^\prime_1$. Note that the right-hand side of this derivation contains an edge labeled by $(A,f_1,g_1)$. Since $G_1=H$, we see that $g_1$ is the identity function $id$. Thus we can apply the type II rule of the form $\nu_{2,\rho_1}$ whose left-hand side equals $(A,f_1,id)=(A,f_1,g_1)$. This rule is also designed in such a way that its application is equivalent to substitution of $G_2$ into $R_1$. This means that $A\underset{HGr_1}{\Rightarrow}G^\prime_1\Rightarrow R_1[G_2/\delta(\rho_1)]=G_1$, as required.
		
		Now, if there is an arbitrary derivation where the rule $\rho_i$ is applied ($1\le i\le K$), then we find the largest $\delta$-derivation within the derivation that includes this occurence $\rho_i$ and remodel it as described above. Thus, the above reasonings complete the proof for one of the inclusions.
		
		To prove the reverse inclusion, it suffices to show how to transform a branch of a derivation containing nonterminals of the form $(A,f,g)$ into a derivation in $HGr$. Let
		\begin{equation*}
		\begin{split}
			A\underset{\nu_{1,\gamma,f_m,g_m}}{\Rightarrow}
			G_{m+1}
			\underset{\nu_{3,\rho_m,f_{m},g_m,f_{m-1},g_{m-1}}}{\Rightarrow}
			G_m
			\underset{\nu_{3,\rho_{m-1},f_{m-1},g_{m-1},f_{m-2},g_{m-2}}}{\Rightarrow}\dots
			\underset{\nu_{3,\rho_2,f_2,g_2,f_1,g_1}}{\Rightarrow} 
			G_2 \underset{\nu_{2,\rho_{m}}}{\Rightarrow} G_1
		\end{split}
		\end{equation*}
		be such a branch (it has to be of a similar form). Then we can consider a derivation in $HGr$ of the form $A\underset{\rho_1}{\Rightarrow} H_1\underset{\rho_2}{\Rightarrow} \dots \underset{\rho_m}{\Rightarrow} H_m \underset{\gamma}{\Rightarrow} H$; the first part of the proof of this theorem allows us to conclude that $H=G_1$. Note that functions $f_i$ and $g_i$ in the first derivation can be arbitrary (they only have to satisfy the condition of type III rules); however, it is not hard to show that they are the same as the functions $f_i$ and $g_i$ that are built on the basis of the second derivation.
		
		This finishes the proof.
	\end{proof}
	There are two important observations regarding the procedure and the lemma above: 
	\begin{enumerate}
		\item After such a procedure there are no recursive $A$-productions;
		\item For each nonterminal $(A,f,g)$ there is no rule $\pi$ such that $\mu(\pi)=(A,f,g)$.
	\end{enumerate}
	
	Note that $HGr_1$ does not preserve all the properties that $HGr$ had. Namely, rules of Type II can contain only one edge. However, our further actions do not affect type II and type III rules so this does not matter. Besides, useless symbols can appear, which also does not bother us (we deleted them only to perform Transformations \ref{edgeless} and \ref{1edged}). This finishes the second step.
	
	\begin{example}\label{app_ex}
		Let a grammar contain two $A$-productions: the first one is recursive and the second one is not recursive. They are depicted below; $\delta$ of each production is represented by the red color.
		\\
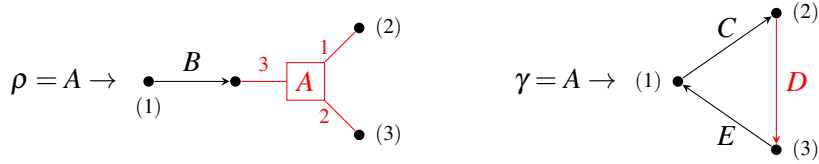
		\begin{tikzpicture}
		\node[] (Start) {$\rho=A\to$};
		\node[node, right=2mm of Start,label=below:{\scriptsize $(1)$}] (N1) {};
		\node[node,right=10mm of N1] (N2) {};
		\node[hyperedge,color=red,right=6mm of N2] (E1) {$A$};
		\node[node,above right=4mm and 4mm of E1,label=right:{\scriptsize $(2)$}] (N3) {};
		\node[node,below right=4mm and 4mm of E1,label=right:{\scriptsize $(3)$}] (N4) {};
		
		\draw[->,black] (N1) -- node[above] {$B$} (N2);
		\draw[-,red] (N2) -- node[above] {\scriptsize 3} (E1);
		\draw[-,red] (N3) -- node[left] {\scriptsize 1} (E1);
		\draw[-,red] (N4) -- node[left] {\scriptsize 2} (E1);
		
		\node[right = 5cm of Start] (Start2) {$\gamma=A\to$};
		\node[node, right=6mm of Start2,label=left:{\scriptsize $(1)$}] (N21) {};
		\node[node,above right=8mm and 12mm of N21,label=right:{\scriptsize $(2)$}] (N22) {};
		\node[node,below right=8mm and 12mm of N21,label=right:{\scriptsize $(3)$}] (N23) {};
		
		\draw[->,black] (N21) -- node[above] {$C$} (N22);
		\draw[->,red] (N22) -- node[right] {$D$} (N23);
		\draw[->,black] (N23) -- node[below] {$E$} (N21);
		\end{tikzpicture}
		\\
		In this grammar the following $\delta$-derivation is possible:
		\\
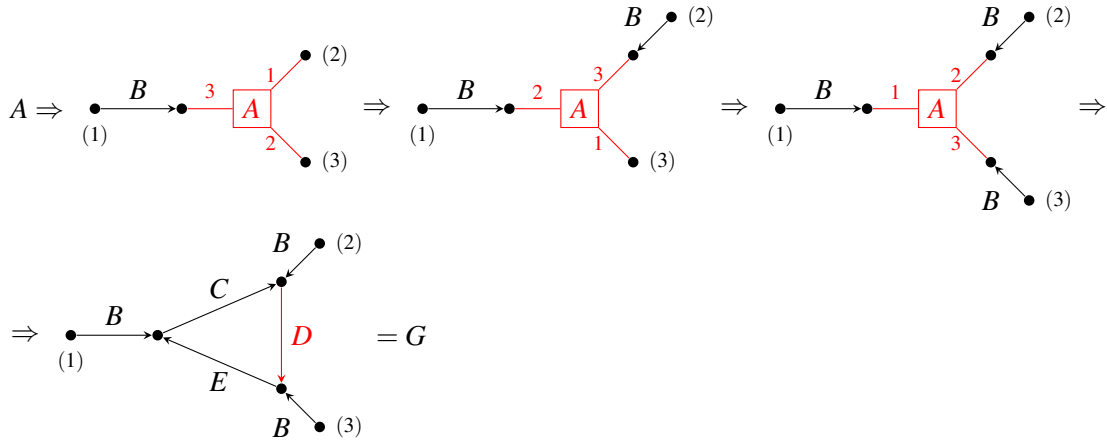
		\begin{tikzpicture}
		\node[] (Start) {$A\Rightarrow$};
		\node[node, right=2mm of Start,label=below:{\scriptsize $(1)$}] (N1) {};
		\node[node,right=10mm of N1] (N2) {};
		\node[hyperedge,color=red,right=6mm of N2] (E1) {$A$};
		\node[node,above right=4mm and 4mm of E1,label=right:{\scriptsize $(2)$}] (N3) {};
		\node[node,below right=4mm and 4mm of E1,label=right:{\scriptsize $(3)$}] (N4) {};
		
		\draw[->,black] (N1) -- node[above] {$B$} (N2);
		\draw[-,red] (N2) -- node[above] {\scriptsize 3} (E1);
		\draw[-,red] (N3) -- node[left] {\scriptsize 1} (E1);
		\draw[-,red] (N4) -- node[left] {\scriptsize 2} (E1);
		
		\node[right = 3.7cm of Start] (Start2) {$\Rightarrow$};
		\node[node, right=2mm of Start2,label=below:{\scriptsize $(1)$}] (N21) {};
		\node[node,right=10mm of N21] (N22) {};
		\node[hyperedge,color=red,right=6mm of N22] (E21) {$A$};
		\node[node,above right=4mm and 4mm of E21] (N23) {};
		\node[node,below right=4mm and 4mm of E21,label=right:{\scriptsize $(3)$}] (N24) {};
		\node[node,above right=4mm and 4mm of N23,label=right:{\scriptsize $(2)$}] (N25) {};

		\draw[->,black] (N21) -- node[above] {$B$} (N22);
		\draw[->,black] (N25) -- node[above left] {$B$} (N23);
		\draw[-,red] (N22) -- node[above] {\scriptsize 2} (E21);
		\draw[-,red] (N23) -- node[left] {\scriptsize 3} (E21);
		\draw[-,red] (N24) -- node[left] {\scriptsize 1} (E21);
		
		\node[right = 4.1cm of Start2] (Start3) {$\Rightarrow$};
		\node[node, right=2mm of Start3,label=below:{\scriptsize $(1)$}] (N31) {};
		\node[node,right=10mm of N31] (N32) {};
		\node[hyperedge,color=red,right=6mm of N32] (E31) {$A$};
		\node[node,above right=4mm and 4mm of E31] (N33) {};
		\node[node,below right=4mm and 4mm of E31] (N34) {};
		\node[node,above right=4mm and 4mm of N33,label=right:{\scriptsize $(2)$}] (N35) {};
		\node[node,below right=4mm and 4mm of N34,label=right:{\scriptsize $(3)$}] (N36) {};

		\draw[->,black] (N31) -- node[above] {$B$} (N32);
		\draw[->,black] (N35) -- node[above left] {$B$} (N33);
		\draw[->,black] (N36) -- node[below left] {$B$} (N34);
		\draw[-,red] (N32) -- node[above] {\scriptsize 1} (E31);
		\draw[-,red] (N33) -- node[left] {\scriptsize 2} (E31);
		\draw[-,red] (N34) -- node[left] {\scriptsize 3} (E31);
		
		\node[right = 4.1cm of Start3] (Start4) {$\Rightarrow$};
		\end{tikzpicture}
		\\
		\begin{tikzpicture}
		\node[] (Start3) {$\Rightarrow$};
		\node[node, right=2mm of Start3,label=below:{\scriptsize $(1)$}] (N31) {};
		\node[node,right=10mm of N31] (N32) {};
		\node[rectangle, minimum width = 5mm, minimum height = 5mm, inner sep = 0mm,right=6mm of N32] (E31) {};
		\node[node,above right=4mm and 4mm of E31] (N33) {};
		\node[node,below right=4mm and 4mm of E31] (N34) {};
		\node[node,above right=4mm and 4mm of N33,label=right:{\scriptsize $(2)$}] (N35) {};
		\node[node,below right=4mm and 4mm of N34,label=right:{\scriptsize $(3)$}] (N36) {};

		\draw[->,black] (N31) -- node[above] {$B$} (N32);
		\draw[->,black] (N35) -- node[above left] {$B$} (N33);
		\draw[->,black] (N36) -- node[below left] {$B$} (N34);
		\draw[->,black] (N32) -- node[above] {$C$} (N33);
		\draw[->,red] (N33) -- node[right] {$D$} (N34);
		\draw[->,black] (N34) -- node[below] {$E$} (N32);
		\node[right=4.2cm of Start3] (End) {$=G$};
		\end{tikzpicture}
		\\
		In the new grammar this derivation turns into the following one:
		\\
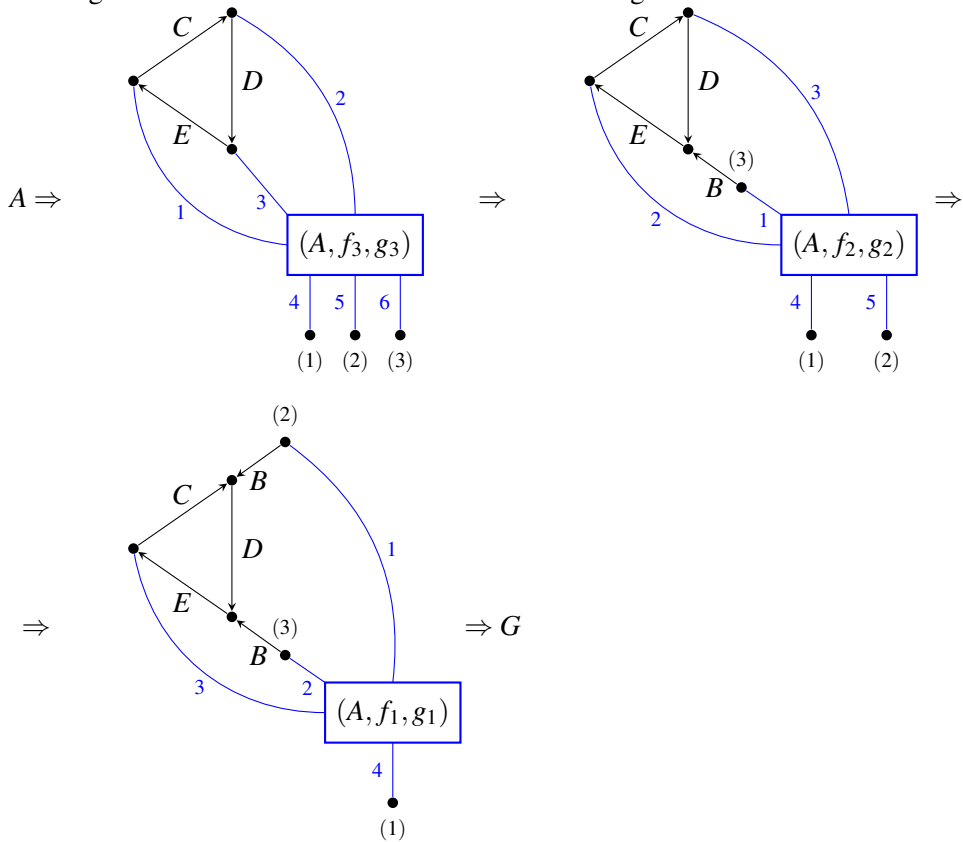
		\begin{tabular}{>{\centering\arraybackslash}m{.1\linewidth} >{\centering\arraybackslash}m{.1\linewidth} >{\centering\arraybackslash}m{.1\linewidth} >{\centering\arraybackslash}m{.1\linewidth} >{\centering\arraybackslash}m{.1\linewidth} >{\centering\arraybackslash}m{.1\linewidth}
				>{\centering\arraybackslash}m{.1\linewidth}
			}
			$A\Rightarrow$ & 
			\begin{tikzpicture}
			\node[node] (N21) {};
			\node[node,above right=8mm and 12mm of N21] (N22) {};
			\node[node,below right=8mm and 12mm of N21] (N23) {};
			
			\draw[->,black] (N21) -- node[above] {$C$} (N22);
			\draw[->,black] (N22) -- node[right] {$D$} (N23);
			\draw[->,black] (N23) -- node[below] {$E$} (N21);
			
			\node[below right = 18mm and 20mm of N21] (Afg3) {$(A,f_3,g_3)$};
			\draw[color=blue, text = black,thick] ($(Afg3)+(-9mm,-4mm)$) rectangle ($(Afg3)+(9mm,4mm)$);
			
			\draw[-,blue] ($(Afg3)+(-9mm,0mm)$) to[bend left=40] node[below] {\scriptsize 1} (N21);
			\draw[-,blue] ($(Afg3)+(-9mm,4mm)$) -- node[below] {\scriptsize 3} (N23);
			\draw[-,blue] ($(Afg3)+(0mm,4mm)$) to[bend right=30] node[right] {\scriptsize 2} (N22);
			
			\node[node, label=below:{\scriptsize $(1)$}] (N24) at ($(Afg3)+(-6mm,-12mm)$) {};
			\node[node, label=below:{\scriptsize $(2)$}] (N25) at ($(Afg3)+(0mm,-12mm)$) {};
			\node[node, label=below:{\scriptsize $(3)$}] (N26) at ($(Afg3)+(6mm,-12mm)$) {};
			\draw[-,blue] ($(Afg3)+(-6mm,-4mm)$) -- node[left] {\scriptsize 4} (N24);
			\draw[-,blue] ($(Afg3)+(0mm,-4mm)$) -- node[left] {\scriptsize 5} (N25);
			\draw[-,blue] ($(Afg3)+(6mm,-4mm)$) -- node[left] {\scriptsize 6} (N26);
			\end{tikzpicture}
			&
			&
			$\Rightarrow$ & 
			\begin{tikzpicture}
			\node[node] (N21) {};
			\node[node,above right=8mm and 12mm of N21] (N22) {};
			\node[node,below right=8mm and 12mm of N21] (N23) {};
			\node[node,below right=4mm and 6mm of N23, label=above:{\scriptsize $(3)$}] (N24) {};
			
			\draw[->,black] (N21) -- node[above] {$C$} (N22);
			\draw[->,black] (N22) -- node[right] {$D$} (N23);
			\draw[->,black] (N23) -- node[below] {$E$} (N21);
			\draw[->,black] (N24) -- node[below] {$B$} (N23);
			
			\node[below right = 18mm and 25mm of N21] (Afg3) {$(A,f_2,g_2)$};
			\draw[color=blue, text = black,thick] ($(Afg3)+(-9mm,-4mm)$) rectangle ($(Afg3)+(9mm,4mm)$);
			
			\draw[-,blue] ($(Afg3)+(-9mm,0mm)$) to[bend left=40] node[below] {\scriptsize 2} (N21);
			\draw[-,blue] ($(Afg3)+(-9mm,4mm)$) -- node[below] {\scriptsize 1} (N24);
			\draw[-,blue] ($(Afg3)+(0mm,4mm)$) to[bend right=30] node[right] {\scriptsize 3} (N22);
			
			\node[node, label=below:{\scriptsize $(1)$}] (N25) at ($(Afg3)+(-5mm,-12mm)$) {};
			\node[node, label=below:{\scriptsize $(2)$}] (N26) at ($(Afg3)+(5mm,-12mm)$) {};
			\draw[-,blue] ($(Afg3)+(-5mm,-4mm)$) -- node[left] {\scriptsize 4} (N25);
			\draw[-,blue] ($(Afg3)+(5mm,-4mm)$) -- node[left] {\scriptsize 5} (N26);
			\end{tikzpicture}
			&
			&
			$\Rightarrow$
			\\
		\end{tabular}
		\\
		\begin{tabular}{>{\centering\arraybackslash}m{.1\linewidth} >{\centering\arraybackslash}m{.1\linewidth} >{\centering\arraybackslash}m{.1\linewidth} >{\centering\arraybackslash}m{.1\linewidth} >{\centering\arraybackslash}m{.1\linewidth}
				>{\centering\arraybackslash}m{.1\linewidth}
			}
			$\Rightarrow$ & 
			\begin{tikzpicture}
			\node[node] (N21) {};
			\node[node,above right=8mm and 12mm of N21] (N22) {};
			\node[node,below right=8mm and 12mm of N21] (N23) {};
			\node[node,below right=4mm and 6mm of N23, label=above:{\scriptsize $(3)$}] (N24) {};
			\node[node,above right=4mm and 6mm of N22, label=above:{\scriptsize $(2)$}] (N25) {};
			
			\draw[->,black] (N21) -- node[above] {$C$} (N22);
			\draw[->,black] (N22) -- node[right] {$D$} (N23);
			\draw[->,black] (N23) -- node[below] {$E$} (N21);
			\draw[->,black] (N24) -- node[below] {$B$} (N23);
			\draw[->,black] (N25) -- node[below] {$B$} (N22);
			
			\node[below right = 18mm and 25mm of N21] (Afg3) {$(A,f_1,g_1)$};
			\draw[color=blue, text = black,thick] ($(Afg3)+(-9mm,-4mm)$) rectangle ($(Afg3)+(9mm,4mm)$);
			
			\draw[-,blue] ($(Afg3)+(-9mm,0mm)$) to[bend left=40] node[below] {\scriptsize 3} (N21);
			\draw[-,blue] ($(Afg3)+(-9mm,4mm)$) -- node[below] {\scriptsize 2} (N24);
			\draw[-,blue] ($(Afg3)+(0mm,4mm)$) to[bend right=30] node[right] {\scriptsize 1} (N25);
			
			\node[node, label=below:{\scriptsize $(1)$}] (N26) at ($(Afg3)+(0mm,-12mm)$) {};
			\draw[-,blue] ($(Afg3)+(0mm,-4mm)$) -- node[left] {\scriptsize 4} (N26);
			\end{tikzpicture}
			&
			&
			$\Rightarrow G$
			\\
		\end{tabular}
		\\
		Here $f_3(1)=2,f_3(2)=3;f_1=f_2=f_3=g_2; g_3(1)=3;g_1=id$ are all the involved partial functions; it is easy to check that they actually satisfy conditions from definitions. Note that if $n$ is the total number of derivation steps, then $f_i$ specifies a relation between external nodes on the $(n-i)$-th and on the $(n-i+1)$-th steps and it is defined by a corresponding production in the old derivation. In this example, $f_i$ are all defined by the production $\rho$, thus they are the same. $g_i$ defines a relation between external nodes on the $(n-i+1)$-th step and external nodes at the end of a derivation (on the $n$-th step). As expected, $g_1$ is the identity function since $n=4$ and $(n-1+1)=n$, so $g_1$ specifies a relation between external nodes on the last and on the last step.
	\end{example}
	
\subsection{Eliminating Recursion and Convertion into the WGNF}
	The rest of the steps are similar to those in the string case. Let us start with $HGr$; let $N=\{A_1,\dots,A_m\}$. Then the following transformation is done.
	\begin{transform}[eliminating recursion]\label{elim_rec}
		\leavevmode
		\\
	\textbf{Input:} an HRG $HGr$ as at the beginning of the proof (without useless symbols or edgeless or chain productions).
	\\
	\textbf{Output:} an equivalent HRG $HGr^\prime$ with the same properties such that for each $\pi\in P$ with $lhs(\pi)=A_i$ either $\mu(\pi)=A_j$ and $i<j$ or $\mu(\pi)\in\Sigma$.
	\end{transform}
	\noindent
	\textbf{Method.}
	\begin{enumerate}
		\item Set $i=1$.
		\item\label{step_elim} Eliminate all the recursive $A_i$-productions according to Lemma \ref{lem_rec}. Now it is argued that if $\pi$ belongs to the set of productions and $lhs(\pi)=A_i$, then $\mu(\pi)=A_k$ for $k>i$ or $\mu(\pi)$ is terminal.
		\item If $i=m$, then a desired grammar is obtained. Otherwise, set $i=i+1$ and $j=1$.
		\item\label{step_j} Let $A_j\to H_1,\dots, A_j\to H_h$ be all the $A_j$ productions. Let $\pi=A_i\to G(e_0:A_j)$ be a production with $\delta(\pi)=e_0$. Then we can replace $\pi$ by the rules $A_i\to G[H_1/e_0], \dots,A_i\to G[H_h/e_0]$ without changing the generated language (see Lemma \ref{lemma_repl}). For these rules we set $\delta(A_i\to G[H_k/e_0])=\delta(A_j\to H_k),\;k=1,\dots,h$. It will now be the case (due to the previous steps of this transformation) that for each $A_j$-production $\pi$ $\mu(\pi)$ equals either a terminal or $A_k$ for $k>j$. Thus all the $A_i$-productions will have this property after such replacement.
		\item If $j=i-1$, go to step \ref{step_elim}; otherwise set $j=j+1$ and go to step \ref{step_j}.
	\end{enumerate}
	\qed
	
	Let $HGr^\prime=\langle N^\prime, \Sigma, P^\prime, S\rangle$ be a grammar that is obtained from $HGr$ after Transformation \ref{elim_rec}. Let $|N^\prime|=M$ (note that $M$ is really large). We already know that $HGr^\prime$ does not have recursions, i.e.~there are no $\delta$-derivations $A\underset{\pi_1}{\Rightarrow} H_1\underset{\pi_2}{\Rightarrow} \dots \underset{\pi_k}{\Rightarrow} H_k$ such that $\mu(\pi_k)=A$. Indeed, due to the second observation after Lemma \ref{lem_rec}: nonterminals of the form $(B,f,g)$ cannot participate in recursive $\delta$-derivations since they do not appear as values of $\mu$. For the remaining nonterminals this is a consequence of the result of Transformation \ref{elim_rec}. Thus we can put a linear order on $N^\prime$ in such a way that for $A\underset{\pi}{\Rightarrow} G(e_0:B)\in P^\prime$, if $\delta(\pi)=e_0$, then either $B$ is terminal or it is nonterminal and $A<B$.
	
	We define the final set of productions $\overline{P}$ as follows: $A\to G$ belongs to $\overline{P}$ if there is a $\delta$-derivation $A\underset{\pi_1}{\Rightarrow} H_1\underset{\pi_2}{\Rightarrow} \dots \underset{\pi_k}{\Rightarrow} H_k=G$ of the length $1\le k\le M$ such that $\mu(\pi_k)$ is terminal. Obviously, the last requirement yields that $G$ has at least one terminal edge.
	\begin{lemma}
		$\overline{HGr}=\langle N^\prime, \Sigma, \overline{P}, S\rangle$ is equivalent to $HGr^\prime$.
	\end{lemma}
	\begin{proof}
		Clearly, $L(\overline{HGr})\subseteq L({HGr}^\prime)$, because productions in $\overline{HGr}$ are composed of that in ${HGr}^\prime$. The reverse inclusion is proved by induction on the length $n$ of derivation in $HGr^\prime$ in more general form: $\langle N^\prime, \Sigma, \overline{P}, X\rangle$ is equivalent to $\langle N^\prime, \Sigma, P^\prime, X\rangle$ for each $X\in N^\prime$.
		\\
		\textbf{Basis. } If $X\to G\in P^\prime$ then it is also in $\overline{P}$ since $G$ is terminal.
		\\
		\textbf{Step. } Let $X\to H(e_0:Y)\overset{n-1}{\Rightarrow} G$ (where $e_0=\delta(X\to H(e_0:Y))$) be a derivation of a terminal graph $G$. This derivation contains a $\delta$-derivation $X\underset{\pi_1}{\Rightarrow} H_1=H(e_0:Y)\underset{\pi_2}{\Rightarrow} H_2 \underset{\pi_3}{\Rightarrow} H_3 \dots \underset{\pi_k}{\Rightarrow} H_k$ and let $\mu(\pi_i)=X_i$ (we also set $X_0=X$). Then $X_0<X_1<\dots<X_k$, $X_k$ is terminal and all $X_i$ are different; thus $k\le M$ and $X\to H_k\in \overline{P}$. Then the source derivation is rebuilt as follows: we apply $X\to H_k$ to $X$ and then do the same with nonterminal edges in $H_k$ using the induction hypothesis. This completes the proof.
	\end{proof}
	Thus $L(\overline{HGr})=L$. Now we have to overcome the following problem: the weak Greibach normal form requires right-hand sides of productions to have exactly one terminal edge while we can guarantee for $\overline{HGr}$ that there is at least one terminal edge. This problem is easily fixed: for each $a\in\Sigma$ we introduce a \emph{new} nonterminal symbol $T_a$ and we add rules of the form $T_a\to\circledcirc(a)$. Now, if $A\to G\in\overline{P}$ has more than one terminal symbol in the right-hand side, then we choose one of them and replace the other ones by corresponding nonterminals: $a:= T_a$. After this step the modified grammar $\overline{HGr}$ still generates $L$, and it is in the weak Greibach normal form.

\section{Algorithmic Complexity}\label{s_compl}
	In the string case polynomial-time algorithms of convertion into the Greibach normal form are known; one is interested whether they exist in the hypergraph case.
	
	It can be seen that the proof we have presented is horrible from the point of view of algorithmic complexity. The main problem is in the second step: the amount of partial bijections from the set $[1,t]$ to itself is equal to $(C_t^0)^20!+(C_t^1)^21!+(C_t^2)^22!+\dots+(C_t^t)^2t!$, which grows even faster then any exponential function. It is less obvious but another ``heavy'' part of the algorithm is the first step --- eliminating chain productions. One would hope that presence of these slow parts is the problem of our proof but this is not the case.
	\begin{example}
		Consider a grammar $HGr_3=\langle\{S\},\{a\},P_3,S\rangle$ where $P$ has 3 productions ($type(S)=n$, $type(a)=n$):
		\begin{enumerate}
			\item $S\to \langle\{v_1,\dots,v_n\},\{e_0\},att,lab,v_1v_2\dots v_n\rangle$, $lab(e_0)=S$, $att(e_0)=v_2v_1v_3v_4\dots v_n$;
			\item $S\to \langle\{v_1,\dots,v_n\},\{e_0\},att,lab,v_1v_2\dots v_n\rangle$, $lab(e_0)=S$, $att(e_0)=v_2v_3v_4\dots v_nv_1$;
			\item $S\to\circledcirc(a)$.
		\end{enumerate}
		Obviously, this grammar has a size $O(n)$. It generates graphs with one $a$-labeled edge such that attachment nodes of this edge are obtained from external nodes of the graph by a permutation composed of the transposition $(12)$ and of the cycle $(12\dots n)$. Since these permutations are generators of $S_n$, the language $L(HGr_3)$ contains $n!$ hypergraphs. All of them have exactly one hyperedge, so for all $H\in L(HGr_3)$ a production $S\to H$ has to belong to a grammar in the WGNF; therefore, even the number of productions grows faster than any exponent apart from the algorithm involved. 
	\end{example}
	This yields the following
	\begin{proposition}
		There is no exponential-time algorithm that takes an HRG generating an $ibHCFL$ and returns an equivalent HRG in the WGNF.
	\end{proposition}
	However, if we consider only HRGs of order $r$ for some fixed $r$ (i.e.~in which types of symbols are not greater than $r$), then the above example ceases to be a problem as well as the first and the second steps in our proof since the number of partial bijections from $[1;r]$ to itself (and, in particular, the size of $S_r$) is constant when $r$ is fixed.
	
\section{Conclusion}\label{s_concl}
	Our desire to introduce the weak Greibach normal form appeared due to another research about hypergraph basic categorial grammars (we presented a paper \cite{Pshenitsyn20} about them at ICGT-2020); in that work we introduce a new approach to describing graph languages and prove that each grammar in this new formalism can be transformed into an equivalent HRG in the WGNF and vice versa. Theorems \ref{Theorem} and \ref{THEOREM} show that ibHCFL is exactly the class of languages generated by grammars in the weak Greibach normal form; thus we also answer a question about the recognizing power of hypergraph basic categorial grammars --- they also generate exactly ibHCFL. Besides, we suppose that the issue of the WGNF is a fundamental theoretical question, and we have presented an answer in line with those in \cite{Engelfriet92,Jansen11}. Note that the result gives a natural grammar characterization of the language class ibHCFL.
	\vspace{5mm}
	\\
	\textbf{Acknowledgments.} I thank my scientific advisor prof. Mati Pentus and anonymous reviewers for their valuable advices.
	
\bibliographystyle{eptcs}
\bibliography{GNFHRG}
	
\end{document}